\documentclass{article}
\usepackage[utf8]{inputenc}
\usepackage[numbers,sectionbib,sort&compress,comma,square]{natbib}
\usepackage[pdftex]{hyperref}
\usepackage{multirow}
\usepackage{verbatim}
\usepackage[ruled,vlined,linesnumbered]{algorithm2e}
\usepackage[table]{xcolor}
\usepackage{amssymb}
\usepackage{amsthm}
\usepackage{graphicx}
\usepackage[caption=false]{subfig}
\usepackage{authblk}

\newtheorem{probaux}{Problem}[section]
\newtheorem{theorem}{Theorem}
\newtheorem{corollary}[theorem]{Corollary}
\newtheorem{lemma}[theorem]{Lemma}

\newenvironment{prob}[3]{\bigskip\noindent\framebox{\parbox{0.97\textwidth}{\begin{probaux}{\sc
#1}\\{\hspace*{1cm} \bf \sf Instance:} #2\\{\hspace*{1cm} \bf \sf Question:}
#3\end{probaux}}}\bigskip}{}

\begin{document}

\title{Biclique-colouring verification complexity and 
biclique-colouring power graphs\thanks{An extended abstract published in: Proceedings
of Cologne Twente Workshop (CTW) 2012, pp. 134--138.
Research partially supported by FAPERJ--Cientistas do Nosso Estado,
and by CNPq-Universal.}}

\author[1]{H\'elio B. Mac\^edo Filho}
\author[2]{Simone Dantas}
\author[3]{\\ Raphael C. S. Machado}
\author[1]{Celina M. H. Figueiredo}
\affil[1]{COPPE, Universidade Federal do Rio de Janeiro}
\affil[2]{IME, Universidade Federal Fluminense}
\affil[3]{Inmetro --- Instituto Nacional de Metrologia, Qualidade e Tecnologia.}

\date{}

\maketitle

\let\thefootnote\relax\footnotetext{
\hfill\today
}

\begin{abstract}
Biclique-colouring is a colouring of the vertices of a graph in such a way that
no maximal complete bipartite subgraph with at least one edge is monochromatic. 
We show that it is co$\mathcal{NP}$-complete to check whether a given 
function that associates a colour to each vertex is a biclique-colouring, 
a result that justifies the search for structured classes where the 
biclique-colouring problem could be efficiently solved. We consider 
biclique-colouring restricted to powers of paths and powers of cycles.    
We determine the biclique-chromatic number of
powers of paths and powers of cycles. The biclique-chromatic number of a power
of a path $P_{n}^{k}$ is $\max(2k + 2 - n, 2)$ if $n \geq k + 1$ and exactly
$n$ otherwise. The biclique-chromatic number of a power of a cycle $C_n^k$
is at most~3 if $n \geq 2k + 2$ and exactly $n$ otherwise; we additionally determine the powers
of cycles that are 2-biclique-colourable. All proofs are algorithmic and
provide polynomial-time biclique-colouring algorithms for graphs in the
investigated classes.
\end{abstract}

\section{Introduction}
\label{s:introduction}

Let $G=(V,E)$ be a simple graph with order $n=|V|$ vertices and
$m=|E|$ edges.
A \emph{clique} of $G$ is a maximal set of vertices of size at least~2 that
induces a complete subgraph of $G$.
A \emph{biclique} of $G$ is a maximal set of vertices that induces a complete
bipartite subgraph of $G$ with at least one edge.
A \emph{clique-colouring} of~$G$ is a function $\pi$ that associates a colour to
each vertex such that no clique is monochromatic. If the function uses at
most~$c$ colours we say that $\pi$ is a \emph{$c$-clique-colouring}.
A \emph{biclique-colouring} of $G$ is a function $\pi$ that associates a colour
to each vertex such that no biclique is monochromatic. If the function $\pi$
uses at most~$c$ colours we say that $\pi$ is a \emph{$c$-biclique-colouring}.
The \emph{clique-chromatic number} of $G$, denoted by $\kappa(G)$, is the least
$c$ for which $G$ has a $c$-clique-colouring. The \emph{biclique-chromatic
number} of $G$, denoted by $\kappa_B(G)$, is the least $c$ for which $G$ has a
$c$-biclique-colouring.

Both clique-colouring and biclique-colouring have a ``hypergraph colouring
version.'' Recall that a hypergraph $\mathcal{H}=(V,\mathcal{E})$ is an ordered
pair where $V$ is a set of vertices and $\mathcal{E}$ is a set of hyperedges,
each of which is a set of vertices. A colouring of hypergraph
$\mathcal{H}=(V,\mathcal{E})$ is a function that associates a colour to each
vertex such that no hyperedge is monochromatic. Let $G=(V,E)$ be a graph and let
$\mathcal{H}_C(G)=(V,\mathcal{E}_C)$ and $\mathcal{H}_B(G)=(V,\mathcal{E}_B)$
be the hypergraphs in which hyperedges are, respectively, $\mathcal{E}_C=\{K \subseteq
V\mid K \mbox{ is a clique of } G\}$ and $\mathcal{E}_B=\{K \subseteq V\mid
K \mbox{ is a biclique of } G\}$ --- hypergraphs $\mathcal{H}_C(G)$ and
$\mathcal{H}_B(G)$ are called, resp., the \emph{clique-hypergraph} and the
\emph{biclique-hypergraph} of~$G$. A clique-colouring of~$G$ is a colouring of
its clique-hypergraph $\mathcal{H}_C(G)$; a biclique-colouring of~$G$ is a
colouring of its biclique-hypergraph $\mathcal{H}_B(G)$.

Clique-colouring and biclique-colouring are analogous problems in the sense that
they refer to the colouring of hypergraphs arising from graphs.
In particular, the hyperedges are subsets of vertices that are clique (resp. biclique).
The clique is a classical important
structure in graphs, hence it is natural that the clique-colouring problem has
been studied for a long time ---
see~\cite{Bacso,Defossez,Kratochvil,DanielMarx}. 
Only recently the biclique-colouring problem started to be
investigated~\cite{1210.7269}. 

Many other problems, initially stated for cliques, have
their version for bicliques~\cite{MR0409253,MR0065617}, such as
\emph{Ramsey number} and \emph{Tur\'an's theorem}.
The combinatorial game called on-line Ramsey
number also has a version for bicliques~\cite{MR2594965}. Although complexity
results for complete bipartite subgraph problems are mentioned
in~\cite{GareyJohnson} and the (maximum) biclique problem is shown to be
$\mathcal{NP}$-hard in~\cite{Yannakakis}, only in the last decade the (maximal)
bicliques were rediscovered in the context of counting
problems~\cite{Gaspers,Prisner}, enumeration problems~\cite{Dias1,Nourine1}, and
intersection graphs~\cite{MarinaJayme}. 

Clique-colouring and biclique-colouring
have similarities with usual vertex-colouring. A proper vertex-colouring is also
a clique-colouring and a biclique-colouring --- in other words, both the
clique-chromatic number and the biclique-chromatic number are bounded above by
the vertex-chromatic number. Optimal vertex-colourings and clique-colourings
coincide in the case of $K_3$-free graphs, while optimal vertex-colourings and
biclique-colourings coincide in the (much more restricted) case of
$K_{1,2}$-free graphs --- notice that the triangle $K_3$ is the minimal
complete graph that includes the graph induced by one edge ($K_2$), while the
$K_{1,2}$ is the minimal complete bipartite graph that includes the graph
induced by one edge ($K_{1,1}$). But there are also essential differences. Most
remarkably, it is possible that a graph has a clique-colouring (resp.
biclique-colouring), which is not a clique-colouring (resp.
biclique-colouring) when restricted to one of its subgraphs.
Subgraphs may even have a larger clique-chromatic number (resp.
biclique-chromatic number) than the original graph. 

Clique-colouring and biclique-colouring also have similarities on complexity issues. 
It is known~\cite{Bacso} that it is co$\mathcal{NP}$-complete
to check whether a given function that associates a colour to each vertex is a clique-colouring
by a reduction from $3DM$. Later, an alternative $\mathcal{NP}$-completeness
proof was obtained by a reduction from a variation of $3SAT$, in order to
construct the complement of a bipartite graph~\cite{Defossez}. Based on this, we open this
paper providing a corresponding result regarding the biclique-colouring problem:
it is co$\mathcal{NP}$-complete to check whether a given function that associates 
a colour to each vertex is a biclique-colouring. The
co$\mathcal{NP}$-completeness holds even when the input is a $\{C_4, K_4\}$-free
graph.

We select two structured
classes for which we provide linear-time biclique-colouring algorithms:
powers of paths and powers of cycles. The choice of those classes has also
a strong motivation since they have been recently investigated in the context of
well studied variations of colouring problems. 
For instance, for a power of a path~$P_n^k$, its $b$-chromatic number is $n$, if $n \leq k + 1$; $k +
1 + \lfloor \frac{n-k-1}{3} \rfloor$, if $k + 2 \leq n \leq 4k + 1$; or $2k +
1$, if $n \geq 4k + 2$; whereas, for a power of a cycle
$C_{n}^{k}$, its $b$-chromatic number is $n$, if $n \leq 2k + 1$; $k + 1$, if $n
= 2k + 2$; at least $\min(n-k-1, k + 1 + \lfloor \frac{n-k-1}{3} \rfloor)$, if $2k + 3
\leq n \leq 3k$; $k + 1 + \lfloor \frac{n-k-1}{3} \rfloor$, if $3k + 1 \leq n
\leq 4k$; or $2k + 1$, if $n \geq 4k + 2$~\cite{MR1979111}. Moreover, other well
studied variations of colouring problems when restricted to powers of cycles have been
investigated: chromatic number~\cite{MR1974376}, chromatic
index~\cite{meidanis}, total chromatic number~\cite{MR2303972}, choice
number~\cite{MR1974376}, and clique-chromatic number~\cite{MR2570638}.
It is known, for a power of a cycle $C_{n}^{k}$, that
the chromatic number and the choice number are both $k + 1 + \lceil r/q
\rceil$, where $n = q(k+1) + t$ with $q \geq 1$, $0 \leq t \leq k$ and $n \geq
2k + 1$, that the chromatic index is the maximum degree of $C_{n}^{k}$ if, and only if,
$n$ is even, that the total chromatic number is at most the maximum degree of
$C_{n}^{k}$ plus 2, when $n$ is even and $n \geq 2k + 1$, and that the
clique-chromatic number is~$2$, when $n \leq 2k + 1$, and is at most~3, when $n
\geq 2k + 2$.
Particularly, in the latter case, the clique-chromatic number is~3, when $n$ is
odd and $n \geq 5$; otherwise, it is~2. Note that total colouring is an open
and difficult problem and remains unsolved for powers of cycles~\cite{MR2303972}. Other
significant works have been done in power
graphs~\cite{MR1454439,MR2423405} and, in particular, in powers of
paths and powers of
cycles~\cite{MR1018529,MR1172679,MR2083449,MR2774114,MR1633075,MR2255625}.

\section{Complexity of biclique-colouring}

The biclique-colouring problem is a variation of the clique-colouring problem. Hence,
it is natural to investigate the complexity of biclique-colouring based on the
tools that were developed to determine the complexity of clique-colouring. 
We show that, similarly to the case of clique-colouring, it is
co$\mathcal{NP}$-complete to check whether a given function that associates a
colour to each vertex of a graph is a biclique-colouring.
To achieve a result in this direction, we prove the $\mathcal{NP}$-completeness
of the following problem: of deciding whether there exists a biclique of a
graph $G$ contained in a given subset of vertices of~$G$. Indeed, a function
that associates a colour to each vertex of a given graph $G$ is a
biclique-colouring if, and only if, there is {\bf no} biclique of~$G$
contained in a subset of the vertices of~$G$ associated with the same colour.

We call {\sc Biclique Containment} the problem that decides whether there exists a
biclique of a graph $G$ contained in a given subset of vertices of~$G$.

\begin{prob}
	{
		Biclique Containment
	}
	{
		Graph~$G = (V, E)$ and $V^\prime \subset V$
	}
	{
		Does there exist a biclique $B$ of $G$ such that $B \subseteq V^\prime$? 
	}
\label{prob:confinamentobicliquemaximal}
\end{prob}

In order to show that {\sc Biclique Containment} is $\mathcal{NP}$-complete, 
we use in Theorem~\ref{thm:bicliquecoloracaoinvalidanpcompleto} a reduction
from {\sc 3SAT} problem.

\begin{theorem}
\label{thm:bicliquecoloracaoinvalidanpcompleto}
	The {\sc Biclique Containment} problem is $\mathcal{NP}$-complete, even if
	the input graph is $\{K_{4}, C_{4}\}$-free.
\end{theorem}

\begin{proof}
Deciding whether a graph has a biclique in a given
subset of vertices is in $\mathcal{NP}$: a biclique is a
certificate and verifying this certificate is trivially polynomial.

We prove that {\sc Biclique Containment} problem is $\mathcal{NP}$-hard by
reducing {\sc 3SAT} to it. The proof is outlined as follows. For every formula
$\phi$, a graph $G$ is constructed with a subset of vertices denoted by
$V^\prime$, such that $\phi$ is satisfiable if, and only if, there exists a
biclique $B$ of $G$ such that $B \subseteq V^\prime$.

	Let $n$ (resp. $m$) be the number of variables (resp. clauses) in formula
	$\phi$. We define the graph $G$ as follows.
	
	\begin{itemize}
		\item For each variable $x_{i}$, $1 \leq i \leq n$, there exist two adjacent
		vertices $x_{i}$ and~$\overline{x_{i}}$. Let $L = \{x_{1}, \dots, x_{n},
		\overline{x_{1}}, \dots, \overline{x_{n}}\}$.
		\item For each clause $c_{j}$, $1 \leq j \leq m$, there exists a vertex 
		$c_{j}$. Moreover, each $c_{j}$, $1 \leq j \leq m$, is adjacent to a vertex $l
		\in \{ x_{1}, \dots, x_{n}, \overline{x_{1}}, \dots, \overline{x_{n}}\}$
		if, and only if, the literal corresponding to $l$ is in the clause
		corresponding to vertex~$c_{j}$. Let $C = \{c_1, \ldots, c_m\}$.
		\item There exists a universal vertex $u$ adjacent to all $x_{i}$,
		$\overline{x_{i}}$, $1 \leq i \leq n$, and to all $c_j$, $1 \leq j \leq m$.
		
	\end{itemize}
	
	We define the subset of vertices $V^\prime$ as $\{u, x_{1}, \dots, x_{n},
	\overline{x_{1}}, \dots, \overline{x_{n}}\}$. Refer to 
	Figure~\ref{fig:cbmnpcompleto} for an example of such construction given a
	formula $\phi = ( x_{1} \vee \overline{x_{2}} \vee x_{4}) \wedge 
	( x_{2} \vee \overline{x_{3}} \vee \overline{x_{5}} )\wedge ( x_{1} \vee
	x_{3} \vee x_{5})$.

\begin{figure}[h]
\center
	\includegraphics[width=\textwidth]{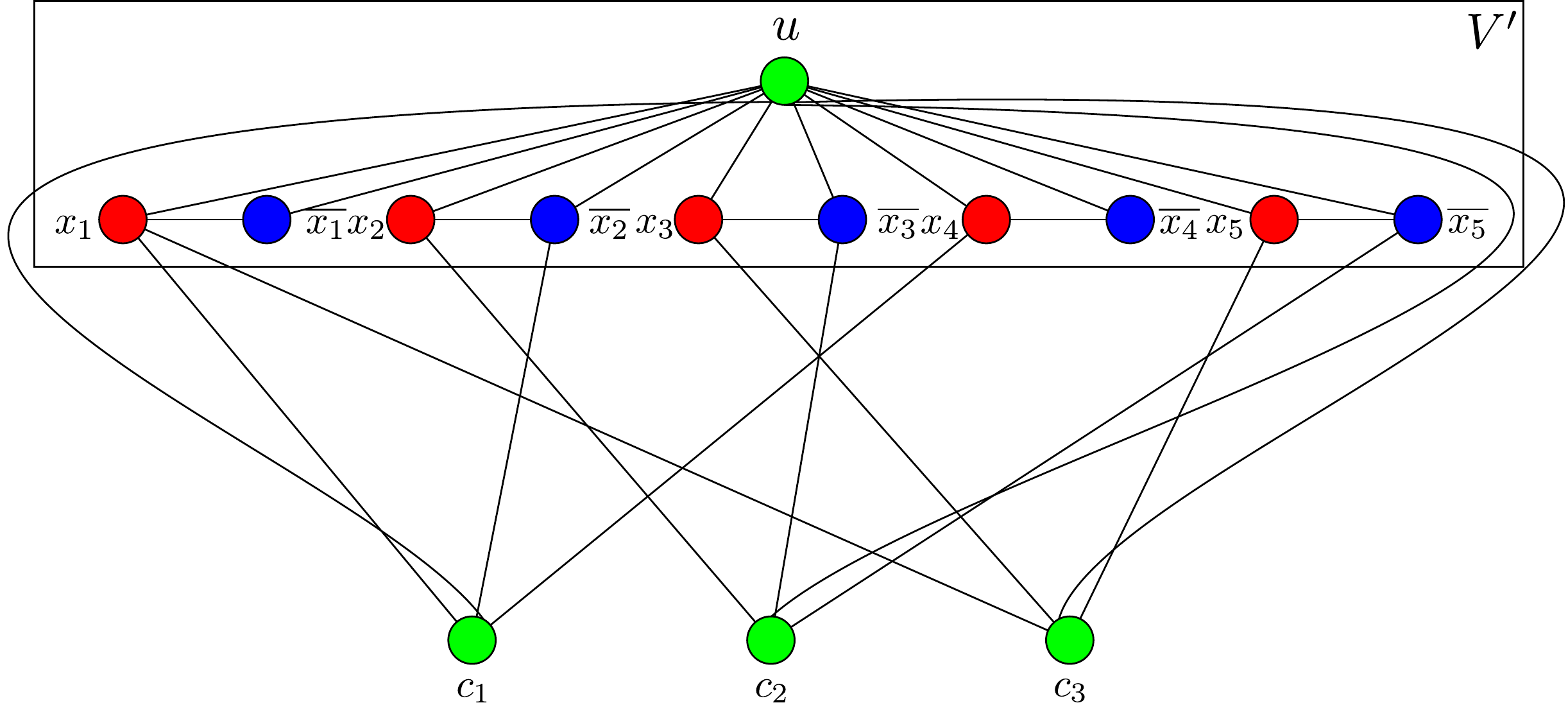}
	\caption{Example for $\phi = ( x_{1} \vee \overline{x_{2}} \vee x_{4}) \wedge 
	( x_{2} \vee \overline{x_{3}} \vee \overline{x_{5}} )\wedge ( x_{1} \vee
	x_{3} \vee x_{5})$ }
	\label{fig:cbmnpcompleto}
\end{figure}
	
	We claim that formula $\phi$ is satisfiable if, and only if, there exists a
	biclique of $G[V^\prime]$ that is also a biclique of~$G$.
	
	Each biclique $B$ of $G[V^\prime]$ containing vertex $u$ corresponds to a
	choice of precisely one vertex of $\{x_i, \overline{x_i}\}$, for each $1
	\leq i \leq n$, and so $B$ corresponds to a truth assignment $v_{B}$ that gives
	true value to variable $x_i$ if, and only if, the corresponding vertex $x_i
	\in B$.
	
	Notice that we may assume three properties on the {\sc 3SAT} instance.
	
	\begin{itemize}
	  \item A variable and its negation do not appear in the same clause. Else,
	  any assignment of values (true or false) to such a variable satisfies the
	  clause.
	  \item A variable appears in at least one clause. Else,
	  any assignment of values (true or false) to such a variable is
	  indifferent to formula $\phi$.
	  \item Two distinct clauses have at most one literal in common. 
	  Else, we can modify the instance as follows. For each clause 
	  $(l_i, l_j, l_k)$, we replace it by clauses $(l_i, x^\prime_1, x^\prime_2)$, 
	  $(l_j, x^\prime_1, \overline{x^\prime_2})$, $(l_j, \overline{x^\prime_1}, x^\prime_3)$,
	  and $(l_k, \overline{x^\prime_1}, \overline{x^\prime_3})$ with variables $x^\prime_1$,
	  $x^\prime_2$, and $x^\prime_3$. Clearly, the number of variables and clauses created is upper
	  bounded by 7 times the number of clauses in the original instance. 
	  Moreover, the original formula is satisfiable if, and only if, the new
	  formula is satisfiable. 
	\end{itemize}    
 
	We consider the bicliques of $G[V^\prime]$ according to two cases.
	\begin{enumerate}
		\item Biclique $B$ does not contain vertex $u$. Then, the biclique is
		precisely formed by a pair of vertices, say $x_i$ and $\overline{x_i}$,
		where $1 \leq i \leq n$. Now, our assumption says that there exists a $c_j$
		adjacent to one precise vertex in $\{x_i, \overline{x_i}\}$ which implies
		that $B$ is not a biclique of $G$.
		\item Biclique $B$ contains vertex $u$. Then, the biclique is
		precisely formed by vertex $u$ and one vertex of $\{x_{i},
		\overline{x_{i}}\}$, for each $1 \leq i \leq n$. $B$ is a biclique of $G$ if, and only if, for each
		$1 \leq j \leq m$, there exists a vertex $l \in L \cap B$ such that $c_j$ is
		adjacent to $l$, which in turn occurs if, and only if, the truth assignment
		$v_B$ satisfies $\phi$. Therefore, $B$ is a biclique of $G$ if, and only if,
		$v_B$ satisfies $\phi$.
	\end{enumerate}
	
	Now, we still have to prove that $G$ is $\{K_4, C_4\}$-free. 
	
	For the sake of contradiction, suppose that there exists a $K_4$ in $G$, say
	$K$. There are no two distinct vertices of $C$ in $K$, since $C$ is an independent set. There
	are no three distinct vertices of $L$ in $K$, since there is a non-edge
	between two of these three vertices. Hence, $K$ precisely contains vertex $u$,
	one vertex of $C$, and two vertices of $L$. Since $K$ is a complete set, the
	two vertices in $L \cap K$ are adjacent and the vertex of $C \cap K$ is adjacent
	to both vertices of $L \cap K$. This contradicts our assumption that a
	variable and its negation do not appear in the same clause. 
	
	For the sake of contradiction, suppose
	there exists a $C_4$ in $G$, say~$H$. The universal vertex $u$
	cannot belong to $H$. Since $C$ is an independent set, $H$ contains at most
	two vertices of $C$. Now, if $H$ contains two vertices of $C$, then the other
	two vertices of $H$ must be two literals, which contradicts our assumption that
	two distinct clauses have at most one literal in common. Since~$L$
	induces a matching, $H$ is not contained in $L$. Therefore, $H$ contains
	one vertex of $C$ and three vertices of $L$, which by the construction of
	$G$ gives the final contradiction.	
	\end{proof}

\begin{corollary}
\label{cor:checkbicliquecolouring}
Let $G$ be a $\{C_4, K_4\}$-free graph. It is co$\mathcal{NP}$-complete to check
if a colouring of the vertices of $G$ is a biclique-colouring.
\end{corollary}

\section{Powers of paths, powers of cycles, and their bicliques}
\label{sec:powerofcyclesandbicliques}
A \emph{power of a path} $P_n^{k}$, for $k \geq 1$, is a simple
graph with $V(G)= \{v_0,\dots, v_{n-1}\}$ and $\{v_i,v_j\}\in E(G)$ if,
and only if, $|i-j| \leq k$. Note that $P_{n}^{1}$ is the induced path $P_n$ on
$n$ vertices and $P_{n}^{k}$, $n \leq k + 1$, is the complete graph $K_{n}$ on $n$
vertices. 
In a power of a path $P_n^k$, the \emph{reach} of an edge $\{v_i,v_j\}$ is $|i -
j|$. A \emph{power of a cycle} $C_n^{k}$, for $k \geq 1$, is a
simple graph with $V(G)= \{v_0,\dots, v_{n-1}\}$ and $\{v_i,v_j\}\in E(G)$ if,
and only if, $\min\{(j-i)\bmod n,(i-j)\bmod n\} \leq k$. Note that $C_{n}^{1}$
is the induced cycle $C_n$ on $n$ vertices and $C_{n}^{k}$, $n \leq 2k + 1$,
is the complete graph $K_{n}$ on $n$ vertices. 
In a power of a cycle $C_n^k$, we take $(v_0,\dots,v_{n-1})$ to be a
\emph{cyclic order} on the vertex set of $G$ and we always perform arithmetic modulo~$n$ on
vertex indices. The \emph{reach} of an edge $\{v_i, v_j\}$ is $\min\{(i-j)\bmod
n, (j-i)\bmod n\}$. The definition of reach is extended to an induced path to be
the sum of the reach of its edges. A \emph{block} is a maximal set of
consecutive vertices. The \emph{size} of a block is the number of vertices in
the block.

All power graphs considered in the present work contain a polynomial
number of bicliques, a sufficient condition for the {\sc Biclique Containment}
problem to be polynomial. In what follows, we explicitly identify the bicliques
of a power of a path and the bicliques of a power of a cycle. 
We say that a biclique of size~2 is a $P_2$ biclique and that a biclique of size~3
is a $P_3$ biclique. Notice that, for each value of $n$ in the considered
range, every biclique in
Lemmas~\ref{lem:powerofpathsbicliques}~and~\ref{lem:powerofcyclesbicliques} always exists.
We refer to Figure~\ref{fig:l4} to illustrate the distinct biclique structures
for each considered case of non-complete powers of cycles.

\begin{figure}[t]
\centering
	\subfloat
		[Power of a cycle $C_{11}^4$ \newline ($2k + 2 \leq n \leq 3k + 1$)] {
			\includegraphics[scale=0.17]{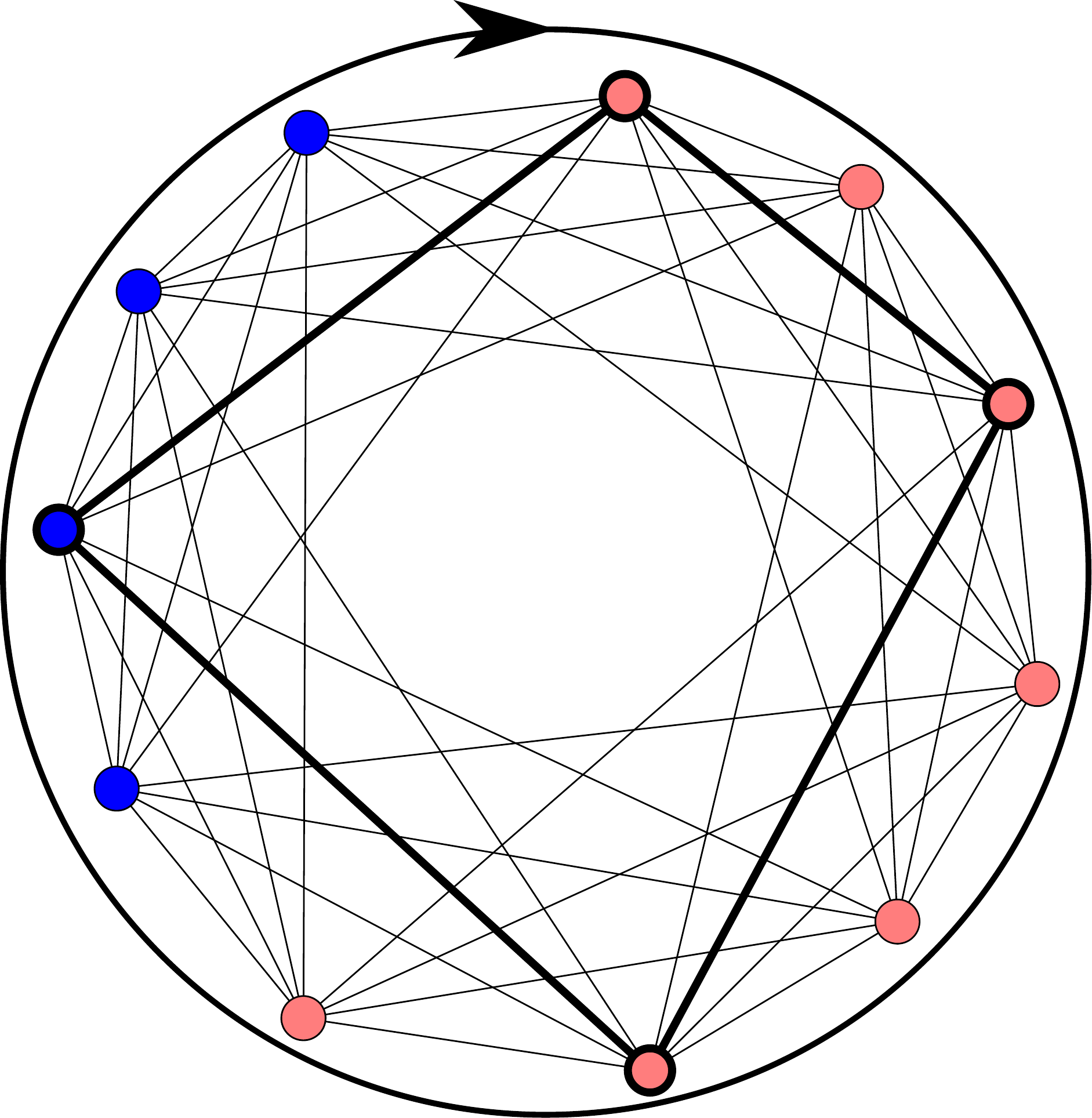}
			\label{fig:c114-l4}
		}
	\qquad
	\subfloat
		[Power of a cycle $C_{11}^3$ \newline ($3k + 2 \leq n \leq 4k$)] {
			\includegraphics[scale=0.17]{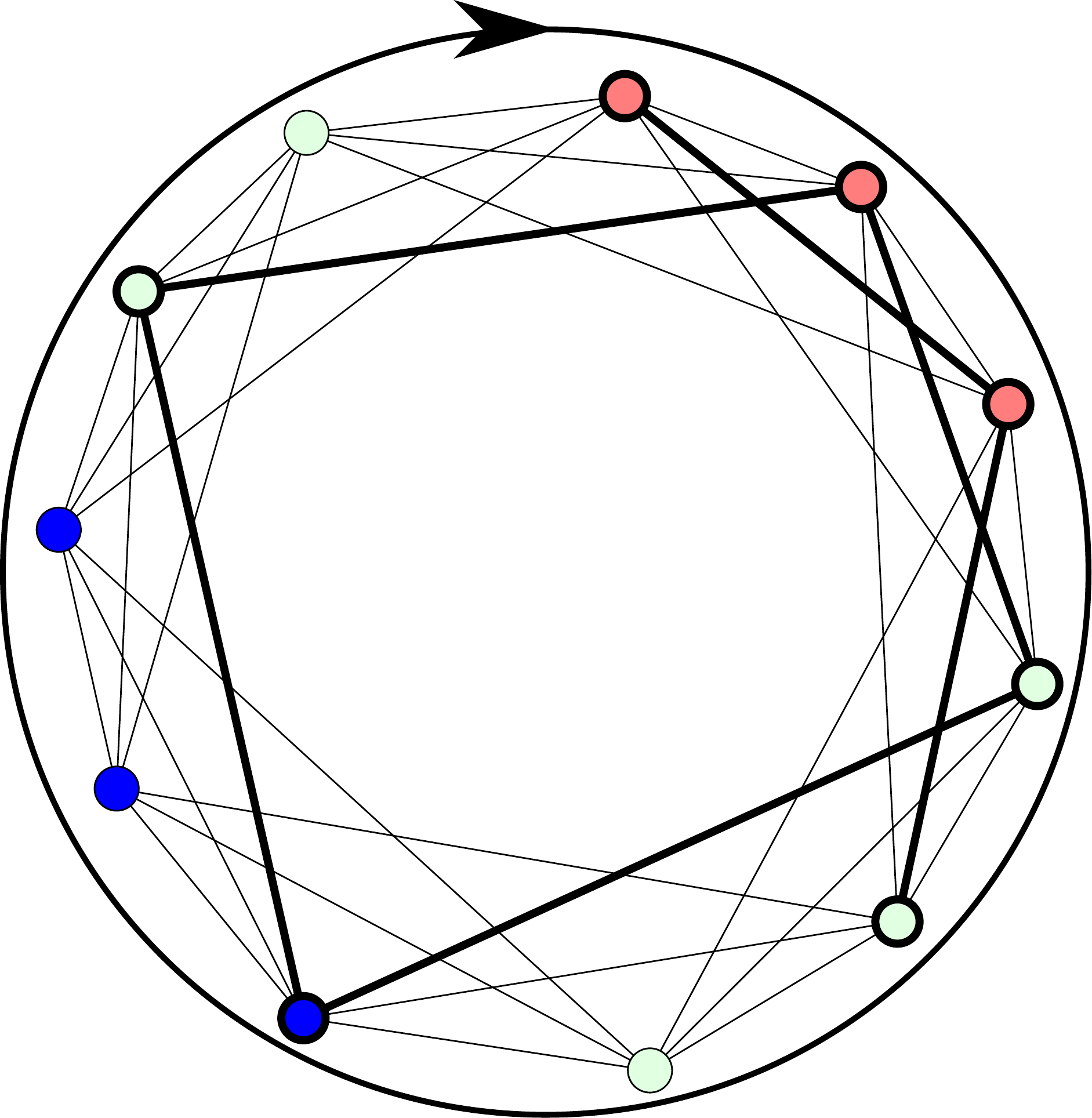}
			\label{fig:c113-l4}
		}
	\qquad
	\subfloat
		[Power of a cycle $C_{11}^2$ \newline ($n \geq 4k + 1$)] {
			\includegraphics[scale=0.17]{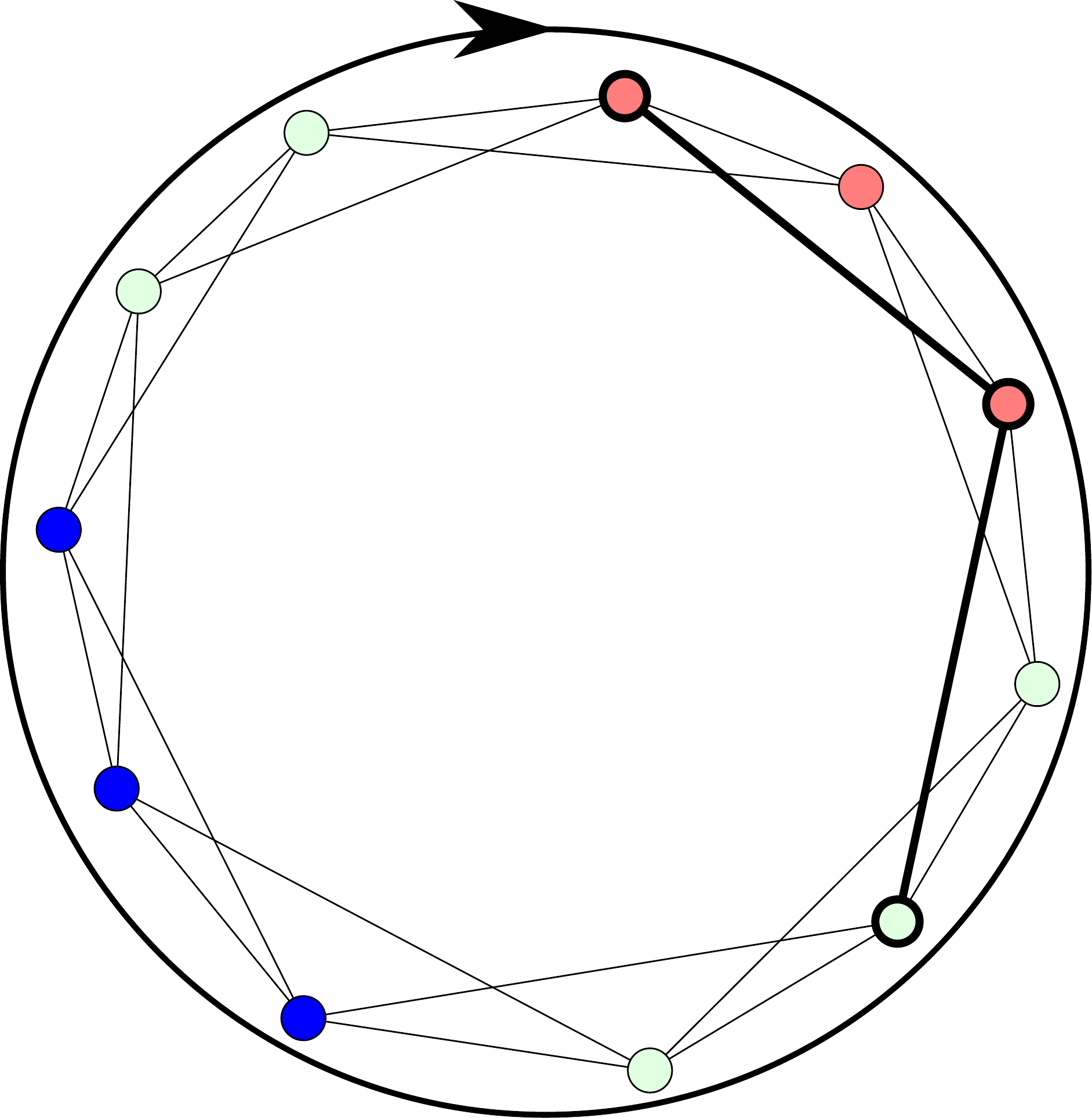}
			\label{fig:c112-l4}
		}
	\caption{For each case of non-complete powers of
	cycles according to Lemma~\ref{lem:powerofcyclesbicliques}, we highlight in
	bold the distinct biclique structures.}
	\label{fig:l4}
\end{figure}
	
\begin{lemma}
\label{lem:powerofpathsbicliques}
 The bicliques of a power of a path $P_n^k$ are precisely: 
 $P_{2}$ bicliques, if $n \leq k + 1$;
 $P_{2}$ bicliques and $P_{3}$ bicliques, if $k + 2 \leq n \leq 2k$; and
 $P_{3}$ bicliques if $n \geq 2k + 1$.
\end{lemma}
\begin{proof}
A power of a path is $K_{1, 3}$-free and $C_4$-free. Thus, the bicliques of a
power of a path are possibly $P_{2}$ or $P_{3}$ bicliques. 

Let $P_{n}^{k}$ be a power of a path with $n
\leq k + 1$. Since $P_n^k = K_n$, every pair of vertices is a $P_2$ biclique.

Let $P_{n}^{k}$ be a power of a path with $k + 2 \leq n \leq
2k$. Since $n > k + 1$ and $k > n - 1 - k$, the edge $\{v_{n-1-k}, v_k\}$ exists
and both vertices $v_{n-1-k}$ and $v_k$ are adjacent to every other vertex of
$P_n^k$. This implies that they define a $P_2$ biclique. Clearly, vertices
$v_0$, $v_k$, and $v_{k+1}$ are distinct and define a $P_3$ biclique.

Now, let $P_{n}^{k}$ be a power of a path with $n \geq
2k + 1$. We claim that always exists only $P_{3}$ biclique. Let $v_{i}$ and
$v_{j}$ be two adjacent vertices in $P_{n}^{k}$, such that $i < j$. If $j \leq
k$, $v_{i}, v_{j}, v_{j + k}$ induce a $P_{3}$, since $v_{i}$ is not adjacent
to $v_{j + k}$. Otherwise $j \geq k + 1$ and $v_{j - (k + 1)}, v_{i}, v_{j}$
induce a $P_{3}$, since $v_{j - (k + 1)}$ is not adjacent to $v_{j}$. We
conclude that every $P_{2}$ is contained in a $P_{3}$, and so every biclique in
$P_{n}^{k}$ is a $P_{3}$ biclique.
\end{proof}

\begin{lemma}
\label{lem:powerofcyclesbicliques}
 The bicliques of a power of a cycle $C_n^k$ are precisely:
  $P_{2}$ bicliques, if $n \leq 2k + 1$;
  $C_4$ bicliques, if $2k + 2 \leq n \leq 3k + 1$;
  $P_{3}$ bicliques and $C_4$ bicliques, if $3k + 2 \leq n \leq 4k$; and
  $P_{3}$ bicliques, if $n \geq 4k + 1$.
\end{lemma}
\begin{proof}
A power of a cycle is $K_{1, 3}$-free. Thus, the bicliques of a power of a cycle
are possibly $P_{2}$, $P_{3}$ or $C_4$ bicliques. Let $C_{n}^{k}$ be a
power of a cycle with $n \leq 2k + 1$. Since $C_n^k = K_n$, every pair of
vertices is a $P_2$ biclique. Otherwise, $n \geq 2k + 2$, and every $P_{2}$ is
properly contained in a $P_{3}$, as we explain next. Let $v_i$ and $v_j$ be two
adjacent vertices in $C_n^k$ such that $i < j$ (indices are taken modulo $n$). Let $v_\ell$ be
the last consecutive vertex after $v_j$ adjacent to $v_i$ along the cyclic
order. It follows that $v_{\ell + 1}$ is not adjacent to $v_i$ but $v_{\ell +
1}$ is adjacent to $v_j$ and that vertices $v_i$, $v_j$, and $v_{\ell + 1}$ define a
$P_3$. Thus, in what follows, each biclique is possibly $P_{3}$ or $C_4$
biclique.

Let $G$ be a power of a cycle $C_{n}^{k}$ with $2k + 2 \leq n \leq 4k$. 
Since $2k + 2 \leq n \leq 4k$, the subset of vertices $H = \{v_0,
v_{\lceil\frac{n}{4} \rceil}, v_{\lceil\frac{n}{2} \rceil},
v_{\lceil\frac{3n}{4} \rceil}\}$ is a $C_4$ biclique. Hence, $G$ has a $C_4$
biclique.

Let $G$ be a power of a cycle $C_{n}^{k}$ with $n \geq 4k + 1$. 
Suppose $P = \{v_{h}, v_{s}, v_{r}\}$ is a~$P_{3}$. If the missing edge is
$\{v_h, v_r\}$, then, by symmetry, we may assume $h < s < r$.
Since $n \geq 4k + 1$, vertices $v_h$ and $v_{r}$ have no common
neighbor with index at most $h - 1$ and at least $r + 1$. Hence, $G$ does not
have a $C_4$ biclique.

Let $G$ be a power of a cycle $C_{n}^{k}$ with $2k + 2 \leq n \leq 3k+1$.
Suppose $P^\prime = \{v_{h^\prime}, v_{s^\prime}, v_{r^\prime}\}$ is a $P_{3}$.
If the missing edge is $\{v_{h^\prime}, v_{r^\prime}\}$, then, by symmetry, we may
assume $h^\prime < s^\prime < r^\prime$. Since $2k + 2 \leq n \leq 3k + 1$,
vertices $v_{h^\prime}$ and $v_{r^\prime}$ have a common neighbor with index at most
$h^\prime - 1$ and at least $r^\prime + 1$ which is not a neighbor of
$v_{s^\prime}$. We conclude that every $P_3$ is contained in a $C_4$ and 
therefore $G$ contains only $C_4$ biclique.

Now, let $G$ be a power of a cycle $C_{n}^{k}$ with $n \geq 3k + 2$. Consider
the $P_{3}$ induced by vertices $v_0$, $v_k$, and $v_{k+1}$. Since $n \geq 3k +
2$, vertices $v_0$ and $v_{k+1}$ have no common neighbor with index at least
$k+2$. Hence, $G$ has a $P_3$ biclique.
\end{proof}

\section{Determining the biclique-chromatic number of~$P_{n}^{k}$}
\label{sec:kappabpowerofpath}

The extreme cases are easy to compute: the densest case occurs when
$n \leq k + 1$, which implies that a power of a path $P_n^k$ is the complete
graph $K_n$ whose biclique-chromatic number is its order $n$, whereas for the non-complete case, the
sparsest case $P_n^k$ occurs when $k = 1$, which implies that a power of a path
$P_n^k$ is the chordless path $P_n$ whose biclique-chromatic number is 2.
According to Lemma~\ref{lem:powerofpathsbicliques}, we consider other two cases:
the less dense case $n \in [k + 2, 2k]$, and the sparse case $n \in [2k + 1,
\infty)$. The proof of
Theorem~\ref{thm:kappabpowerofpathfirstinterval} (resp.
Theorem~\ref{thm:kappabpowerofpathsecondinterval})
additionally yields an efficient $2k+2-n$-biclique-colouring (resp.
2-biclique-colouring) algorithm for the less dense case (resp. for the sparse
case).

\begin{theorem}
\label{thm:kappabpowerofpathfirstinterval}
 A power of a path $P_n^k$, when $k + 2 \leq n \leq 2k$, has
 biclique-chromatic number~$2k + 2 - n$.
\end{theorem}

\begin{proof}
Let $G$ be a power of a path $P_{n}^{k}$ with $k + 2 \leq n \leq 2k$.
Each of the vertices $v_{n-1-k}, \ldots, v_k$ is universal and any pair of
vertices in $\{v_{n-1-k},\ldots,v_k\}$ induces a $P_{2}$ biclique in the
graph. Hence, we are forced to give distinct colours to each of the vertices
$v_{n-1-k}, \ldots, v_k$ and we have $\kappa_{B}(G) \geq 2k + 2 - n$.
 
We define $\pi:V(G)\rightarrow\{1, \ldots, 2k + 2 - n\}$ by giving
(arbitrarily) distinct colours $3, \ldots, 2k + 2 - n$ to vertices $v_{n - k},
\ldots, v_{k - 1}$. Now, use colour $1$ in the uncoloured vertices before $v_{n
- k}$ and colour $2$ in the uncoloured vertices after $v_{k - 1}$. 
Every monochromatic edge contains either both end vertices before $v_{n - k}$ or
both end vertices after $v_{k - 1}$. By symmetry, consider $\{v_i, v_j\}$ a
monochromatic edge such that $i < j < n - k$. Now, vertices $v_i, v_j, v_{j +
k}$ induce a $P_3$ biclique. Since any choice of three vertices either before
$v_{n - k}$ or after $v_{k - 1}$ defines a triangle, $\pi$ is a
biclique-colouring of $G$.

We refer to Figure~\ref{fig:path1} to illustrate the
given $(2k + 2 - n)$-biclique-colouring.
\end{proof}

\begin{theorem}
\label{thm:kappabpowerofpathsecondinterval}
 A power of a path $P_n^k$, when $n \geq 2k + 1$, has
 biclique-chromatic number~2.
\end{theorem}

\begin{proof}
Let $G$ be a power of a path $P_{n}^{k}$ with $n \geq 2k + 1$. 
Let $n = ak + t$, with $0 \leq t < k$. We define
$\pi:V(G)\rightarrow\{blue,red\}$ as follows. A number of $a$
monochromatic-blocks of size $k$ switching colours \emph{red} and \emph{blue}
alternately, followed by a monochromatic-block of size $t$ with \emph{red}
colour if $a$ is even or \emph{blue} colour if $a$ is odd.
We refer to Figure~\ref{fig:path2} to illustrate the given
2-biclique-colouring.

Lemma~\ref{lem:powerofpathsbicliques} says that every biclique of $G$ is a
$P_{3}$. Thus, every biclique is polychromatic, since it contains vertices from
two consecutive monochromatic-blocks (with distinct colours by the
given colouring).
\end{proof}

\begin{figure}[t]
\centering
	\subfloat
		[$(2k + 2 - n)$-biclique-colouring, when $k + 1 \leq n \leq 2k$.] {
			\includegraphics[scale=0.34]{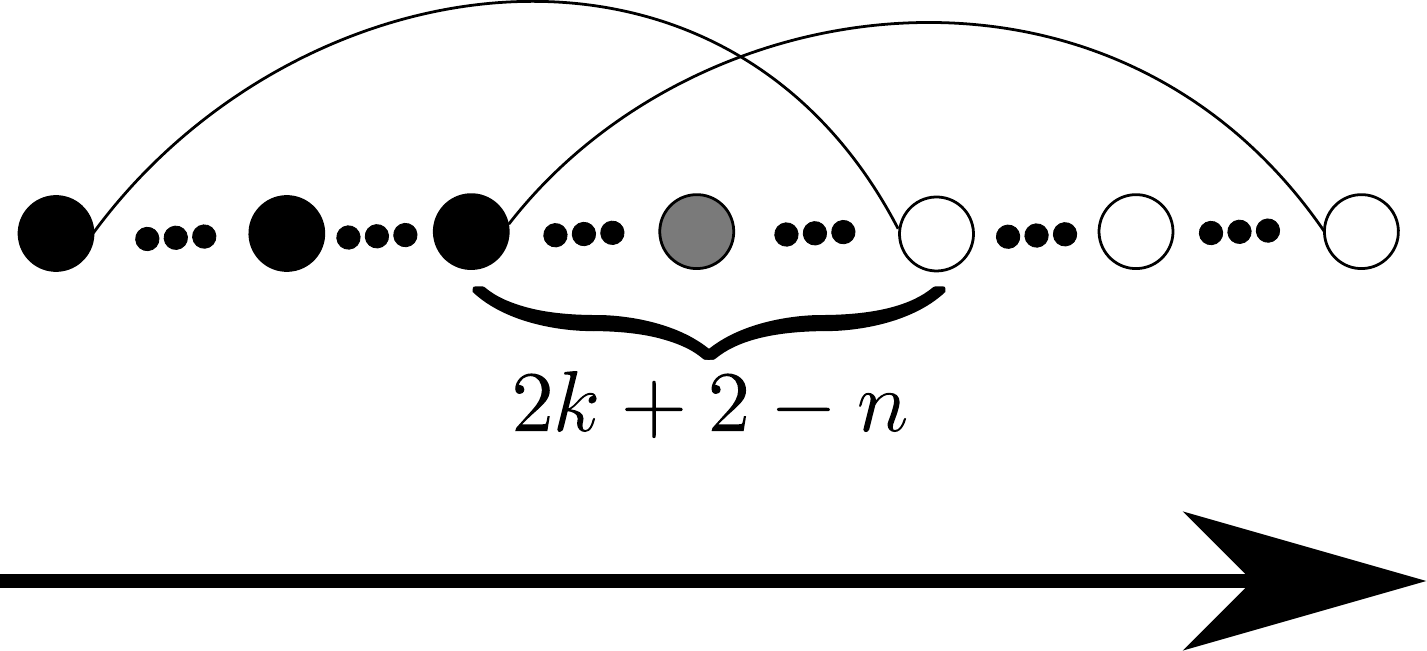}
			\label{fig:path1}
		}
	\qquad
	\subfloat
		[$2$-biclique-colouring, when $n \geq 2k + 2$ and $0 \leq t
		< k$.] {
			\includegraphics[scale=0.34]{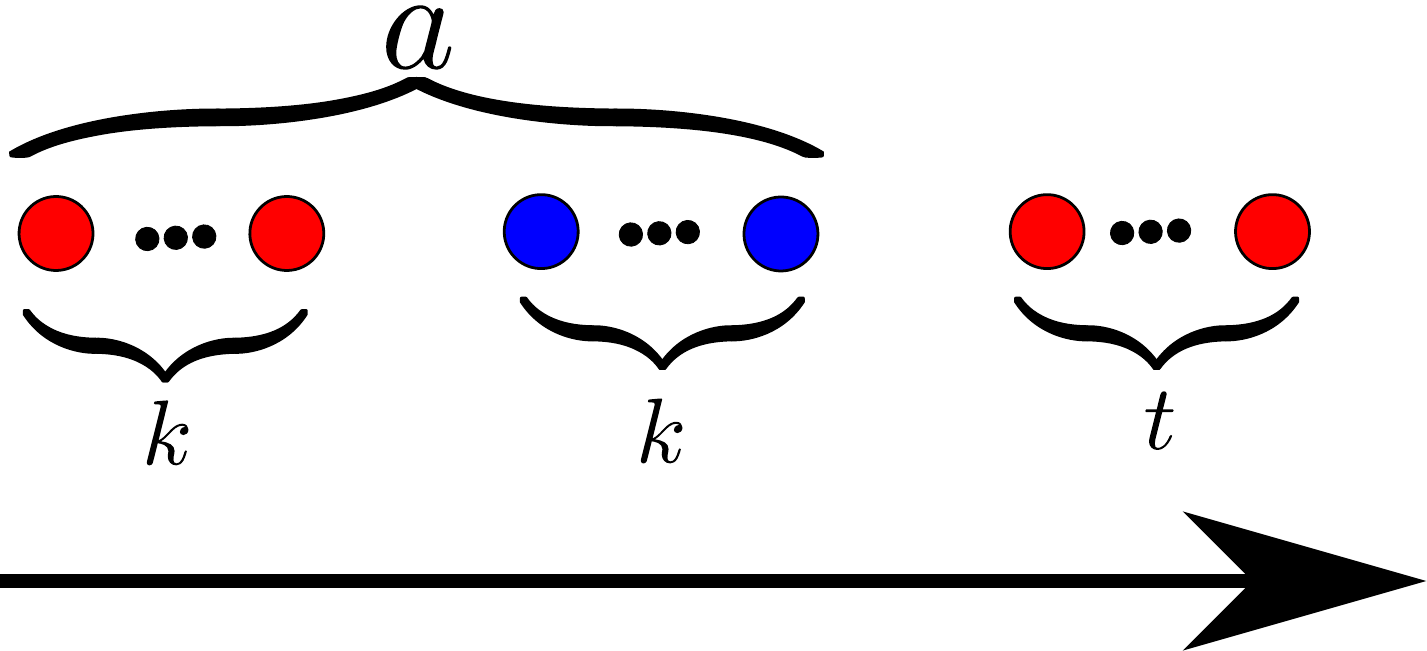}
			\label{fig:path2}
		}
	\caption{Biclique-colouring of
	powers of paths}
	\label{fig:pathbicliquecolouring}
\end{figure}

\section{Determining the biclique-chromatic number of~$C_{n}^{k}$}
\label{sec:bicliqueupperboundpowerofcycle}

The extreme cases are easy to compute: the densest case occurs when
$n \leq 2k + 1$, which implies that a power of a cycle $C_n^k$ is the complete
graph $K_n$ whose biclique-chromatic number is its order $n$, whereas for the
non-complete case, the sparsest case $C_n^k$ occurs when $k = 1$, which implies
that a power of a cycle $C_n^k$ is the chordless cycle $C_n$ whose
biclique-chromatic number is 2.
According to Lemma~\ref{lem:powerofcyclesbicliques}, we consider other two
cases:
the less dense case $n \in [2k + 2, 3k + 1]$, whose biclique-chromatic number is
2, and the sparse case $n \in [3k + 2, \infty)$.

The division algorithm says that any
natural number $a$ can be expressed using the equation $a = bq + t$, with a
requirement that $0 \leq t < b$. We shall use the following version where $b$ is
even and $0 \leq t < 2k$.

\begin{theorem}[Division algorithm]
\label{thm:division}
Given two natural numbers $n$ and $k$,
with $n \geq 2k$, there exist unique natural numbers $a$ and $t$ such that
$n = ak + t$, $a \geq 2$ is even, and $0~\leq~t~<~2k$.
\end{theorem}

Given a non-complete power of a cycle, Lemma~\ref{lem:3colouringnomonoP3} shows
that there exists a 3-colouring of its vertices such that no $P_3$ is
monochromatic. Since every biclique contains a $P_3$, 
Lemma~\ref{lem:3colouringnomonoP3} provides an upper bound of~3 for the 
biclique-chromatic number of a power of a cycle --- the proof of 
Lemma~\ref{lem:3colouringnomonoP3} additionally yields an efficient 
3-biclique-colouring algorithm using the version of the division 
algorithm stated in Theorem~\ref{thm:division}. 
Moreover, this upper bound of~3 to the biclique-chromatic number 
is tight. Please refer to Figure~\ref{fig:c113} for an example of a graph 
not 2-biclique-colourable.

\begin{lemma}
\label{lem:3colouringnomonoP3}
Let $G$ be a power of a cycle $C_n^k$, where $n \geq 2k + 2$. Then, $G$ admits a
3-colouring of its vertices such that $G$ has \textbf{no} monochromatic $P_3$.
\end{lemma}
\begin{proof}
Let $G$ be a power of a cycle $C_{n}^{k}$ with $n \geq 2k + 2$. 
Theorem~\ref{thm:division} says that $n = ak + t$
for natural numbers $a$ and $t$, $a \geq 2$ is even, and $0 \leq t < 2k$. If $0
\leq t \leq k$, we define $\pi:V(G)\rightarrow\{blue,red,green\}$ as follows. An even number
$a$ of monochromatic-blocks of size $k$ switching colours \emph{red} and
\emph{blue} alternately, followed by a monochromatic-block of size $t$ with
colour \emph{green}. Otherwise, i.e. $k < t < 2k$, we define
$\pi:V(G)\rightarrow\{blue,red,green\}$ as follows. An odd number $a+1$ of
monochromatic-blocks of size $k$ switching colours \emph{red} and \emph{blue}
alternately, followed by a monochromatic-block of size $k$ with colour
\emph{green}, a monochromatic-block of size $k$ with colour \emph{blue}, and a
monochromatic-block of size $t - k$ with colour \emph{green}.  We refer to
Figure~\ref{fig:restinho} to illustrate the former 3-biclique-colouring and to
Figure~\ref{fig:restao} to illustrate the latter 3-biclique-colouring.

Consider any three vertices $v_i$, $v_j$ and $v_\ell$ with the
same colour. Then, either they are in the same monochromatic-block --- and
induce a triangle --- or two of them are not in consecutive
monochromatic-blocks -- and induce a disconnected graph. In both cases, $v_i$,
$v_j$ and $v_\ell$ do not induce a $P_3$. 
\end{proof}

\begin{theorem}
\label{thm:bicliqueupperboundpowerofcycle}
A power of a cycle $C_n^k$, when $n \geq 2k + 2$,
has biclique-chromatic number at most~3.
\end{theorem}

\begin{figure}[t]
\centering
	\includegraphics[scale=0.2]{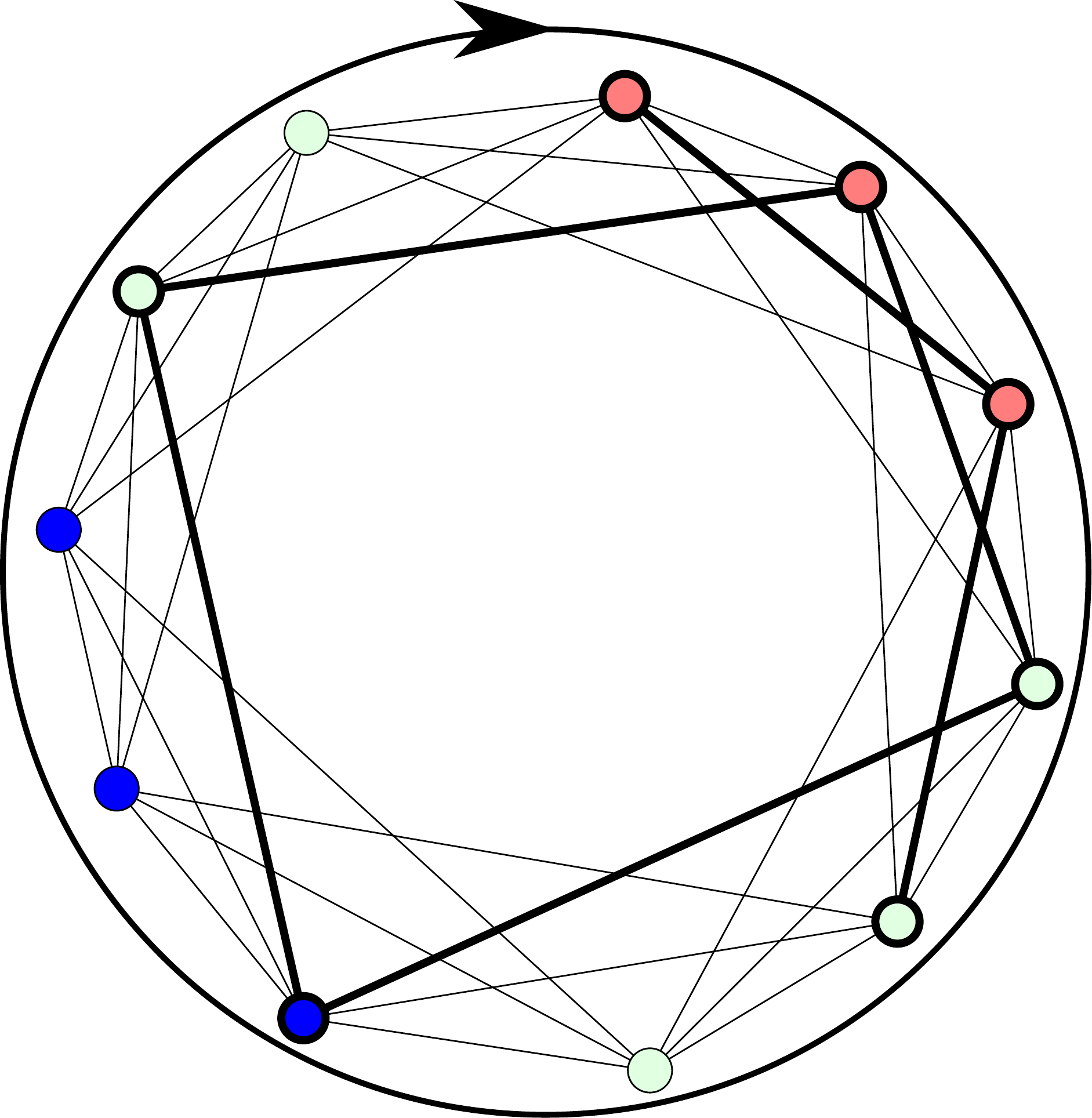}
	\caption{Power of a cycle $C_{11}^{3}$ with biclique-chromatic number~3.
	We highlight in bold a $P_{3}$ biclique of reach $4$ and a $C_4$ biclique.}
	\label{fig:c113}
\end{figure}

\begin{figure}[t]
\centering
	\subfloat
		[3-biclique-colouring, when $n \geq 2k + 2$ and $0 \leq t
		\leq k$.] {
			\includegraphics[scale=0.2]{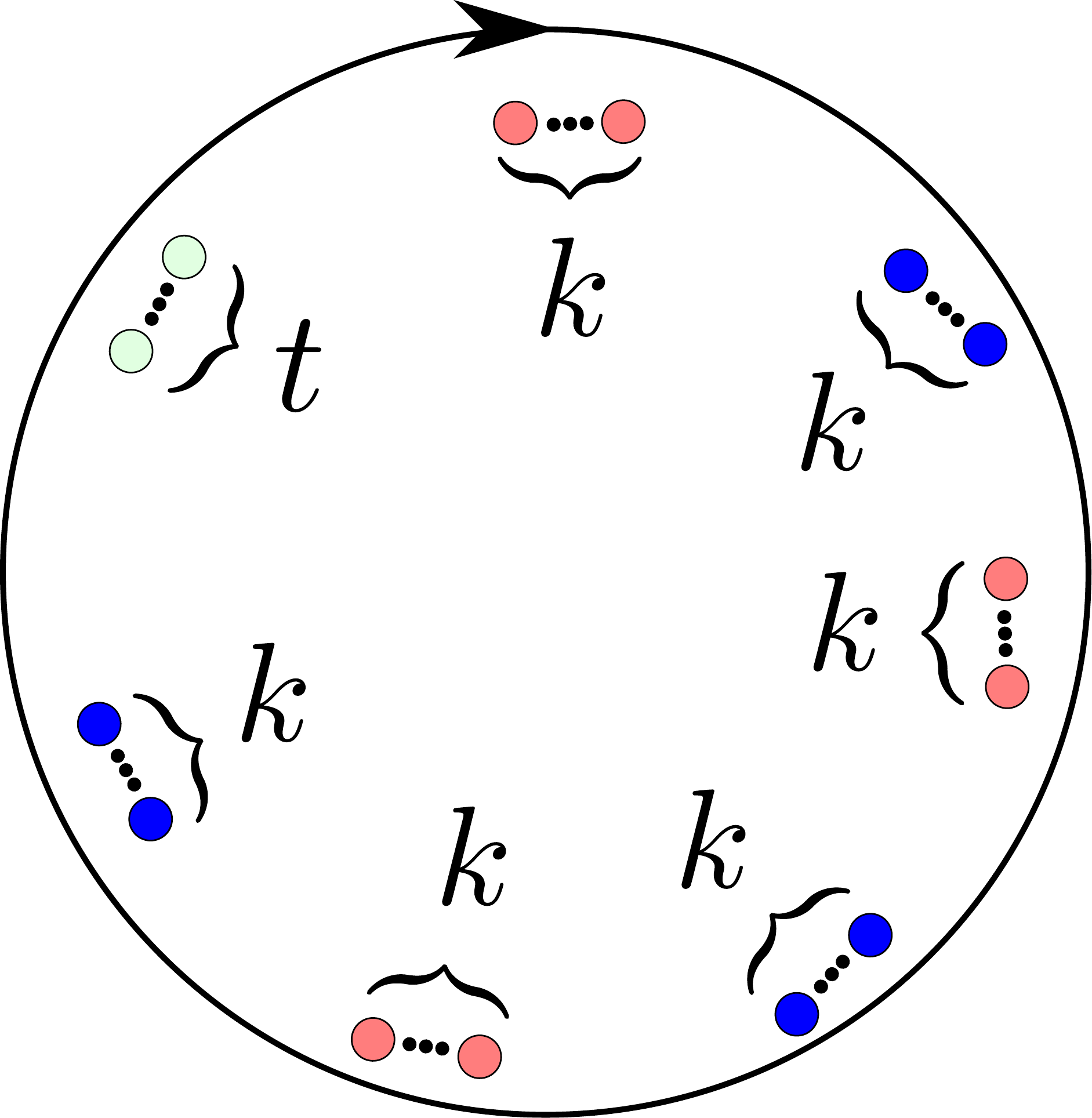}
			\label{fig:restinho}
		}
	\qquad
	\subfloat
		[3-biclique-colouring, when $n \geq 2k + 2$ and $k < t <
		2k$.] {
			\includegraphics[scale=0.2]{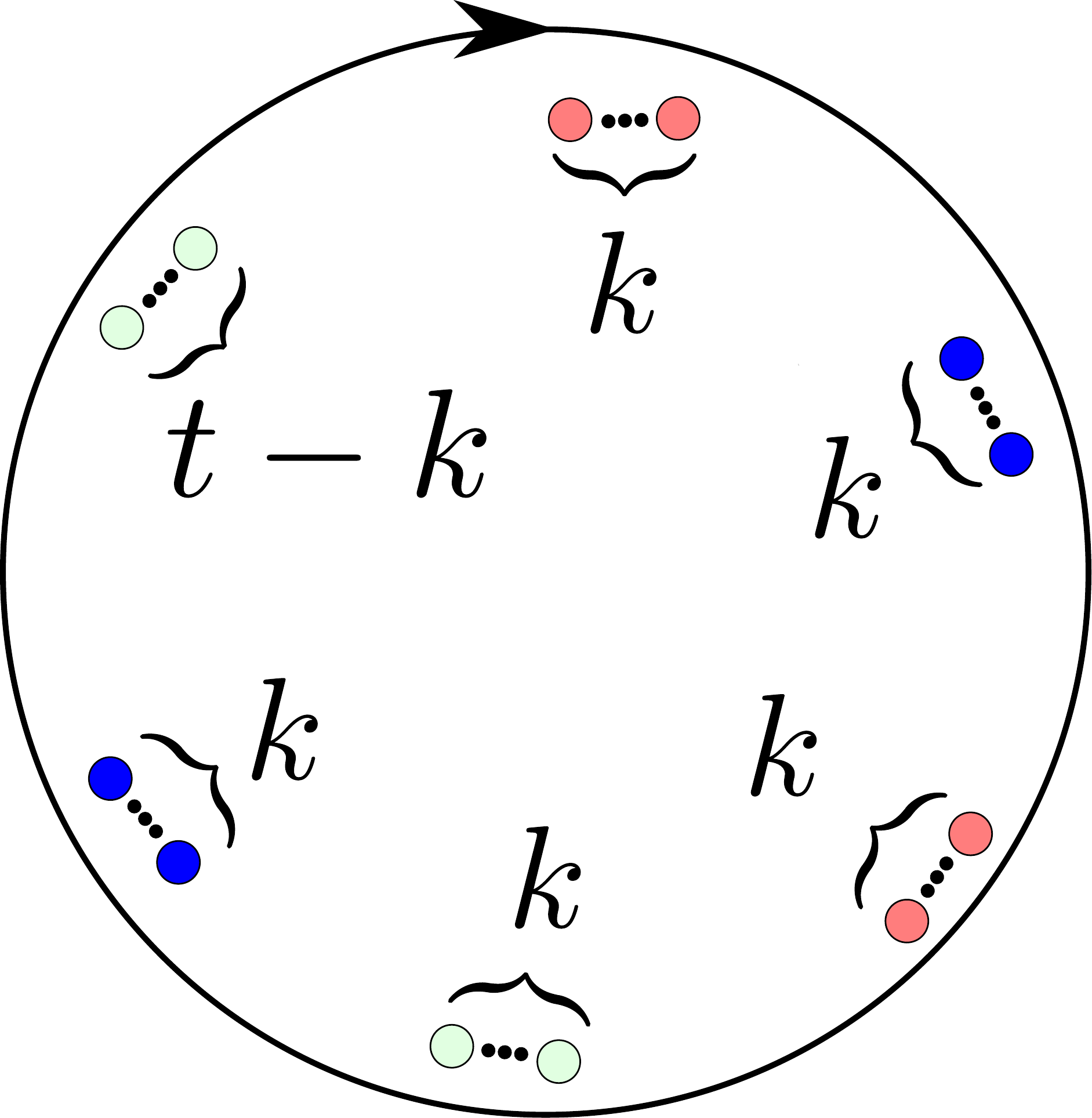}
			\label{fig:restao}
		}
	\qquad
		\subfloat
		[2-biclique-colouring
		when $2k + 2 \leq n \leq 3k + 1$] 
		{
			\includegraphics[scale=0.2]{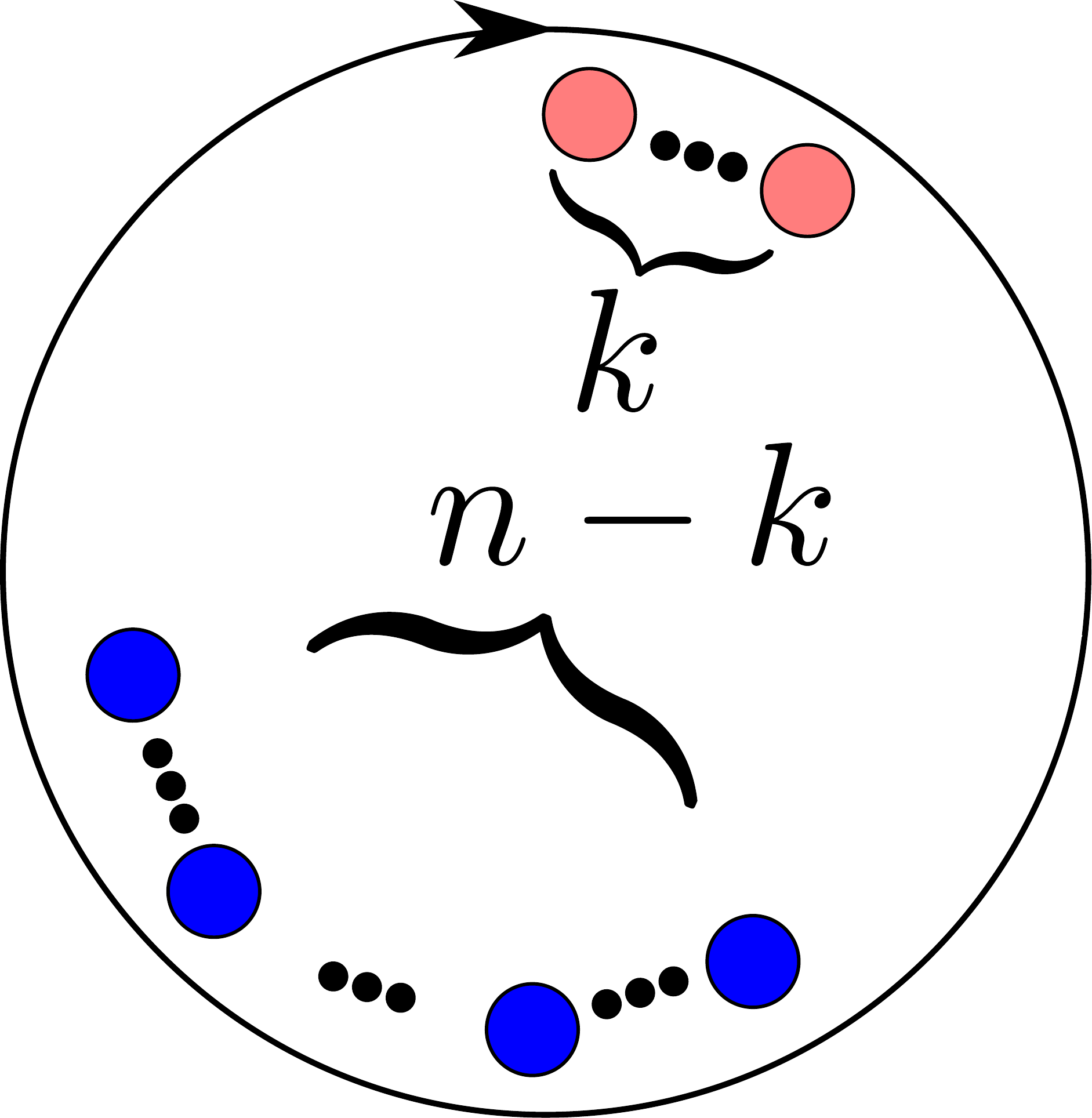}
			\label{fig:excessao}
		}
	\qquad
		\subfloat
		[2-biclique-colouring of a 2-biclique-colourable graph,
		when $n \geq 3k + 2$] 
		{
			\includegraphics[scale=0.2]{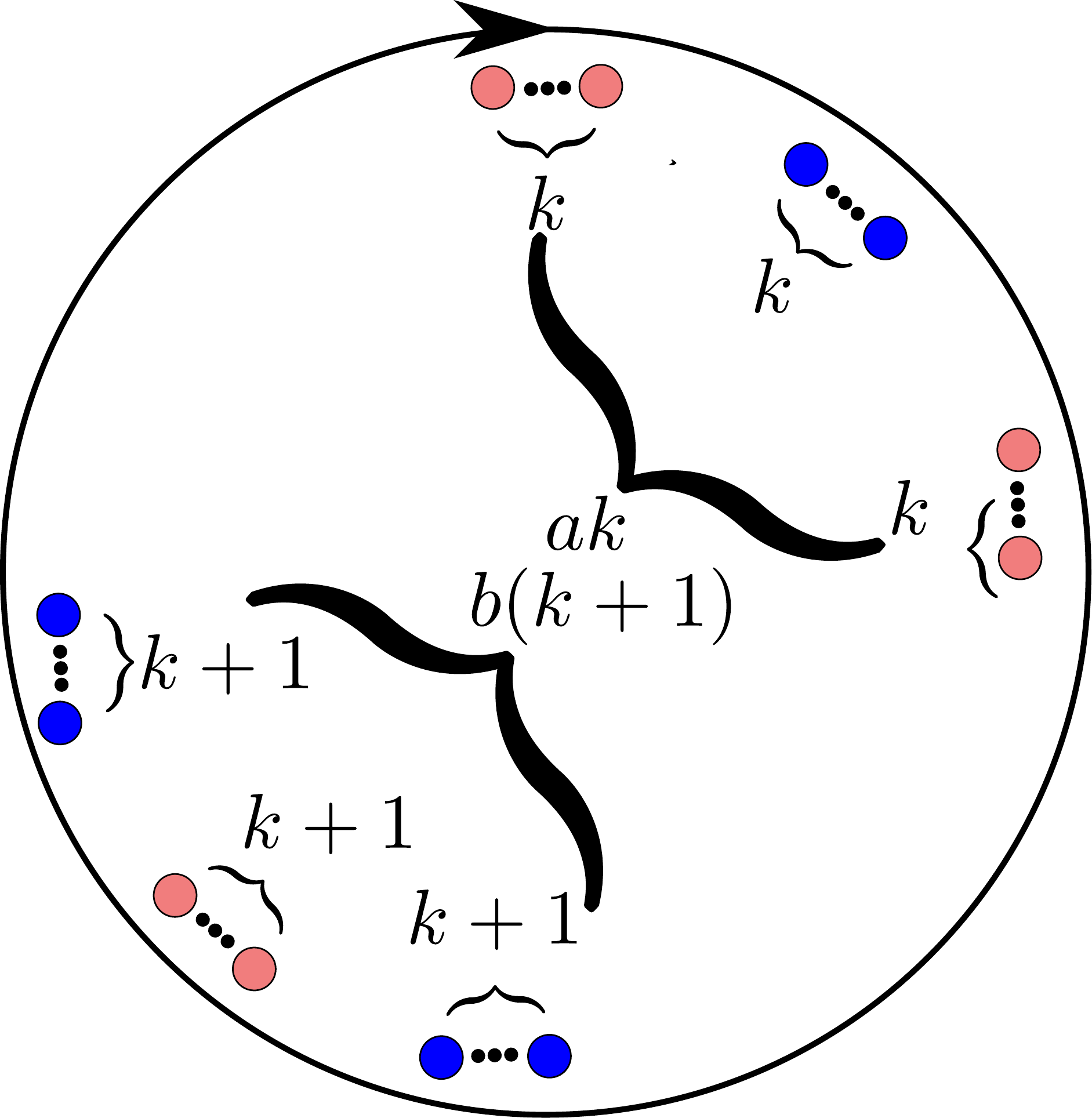}
			\label{fig:semresto}
		}
	\caption{Biclique-colouring of
	powers of cycles}
	\label{fig:2nd3bicliquecolouring}
\end{figure}

As a consequence of Theorem~\ref{thm:bicliqueupperboundpowerofcycle}, every
non-complete power of a cycle has biclique-chromatic number~2 or~3, and it is a
natural question how to decide between the two values.
We first settle this question in the less dense case $n \in [2k + 2, 3k + 1]$.
In fact, we show that all powers of cycles in the less dense case $n \in [2k +
2, 3k + 1]$ are 2-biclique-colourable --- the proof of
Theorem~\ref{thm:kappabpowerofcyclesecondinterval} additionally yields an
efficient 2-biclique-colouring algorithm.

\begin{theorem}
\label{thm:kappabpowerofcyclesecondinterval}
 A power of a cycle $C_n^k$, when $2k + 2 \leq n \leq 3k + 1$, has
 biclique-chromatic number~2.
\end{theorem}

\begin{proof}
Let $G$ be a power of a cycle $C_{n}^{k}$ with $2k + 2 \leq n \leq 3k + 1$. We
define $\pi:V(G)\rightarrow\{blue,red\}$ as follows. A monochromatic-block of
size $k$ with colour \emph{red} followed by a monochromatic-block of size $n-k$
with colour \emph{blue}. We refer to Figure~\ref{fig:excessao} to illustrate the
given 2-biclique-colouring.

Recall that every biclique of $G$ is a $C_{4}$ biclique. For the sake of
contradiction, suppose that there exists a monochromatic set $H$ of four
vertices. If $H$ is contained in the block of size $k$, then $H$ induces a
$K_4$ and cannot be a $C_4$. Otherwise, $H$ is contained in the block of size
$n - k \leq 2k + 1$ and there exists a subset of $H$ which induces a triangle,
so that $H$ cannot be a $C_4$ biclique.
\end{proof}

The sparse case $n \geq 3k + 2$ is more tricky. Let $G$ be a power of a cycle
$C_{n}^{k}$ with $n \geq 3k + 2$. Following
Lemma~\ref{lem:powerofcyclesbicliques}, there always exists a $P_3$ biclique in
$G$. Clearly, a biclique-colouring of $G$ has every $P_3$ biclique
polychromatic, but we may think that there exists some monochromatic $P_3$ (not
biclique). Nevertheless, we prove that $G$ has biclique-chromatic number~2 if,
and only if, there exists a 2-colouring of $G$ such that \textbf{no} $P_3$ is
monochromatic, which happens exactly when there exists a 2-colouring of $G$
where every monochromatic-block has size $k$ or $k + 1$.

\begin{figure}[t]
\centering
	\subfloat
		[vertices $v_{i-1}$, $v_{i+k-x}$, and $v_{i+k}$ induce a monochromatic
		$P_{3}$ with reach $k+1$] {
			\includegraphics[scale=0.3]{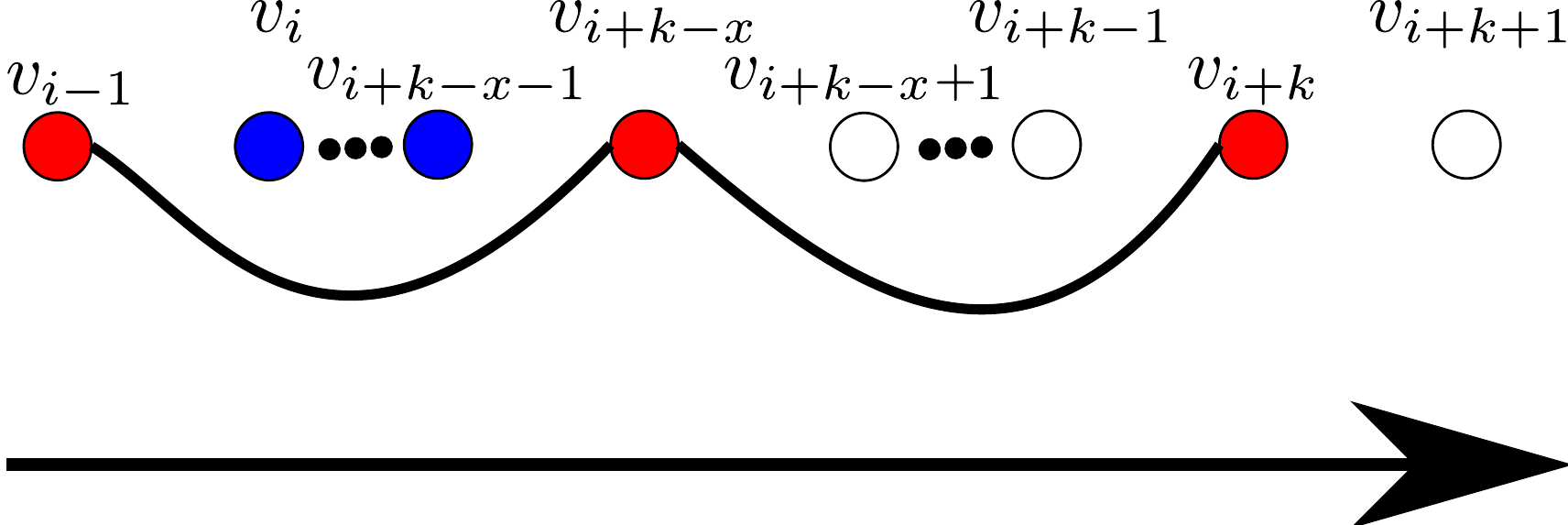}
			\label{fig:catarina1v2}
		}
	\qquad
	\subfloat
		[vertices $v_{i}$, $v_{i+k}$, and $v_{i+k+1}$ induce a monochromatic $P_{3}$
		with reach $k+1$]
		{
			\includegraphics[scale=0.3]{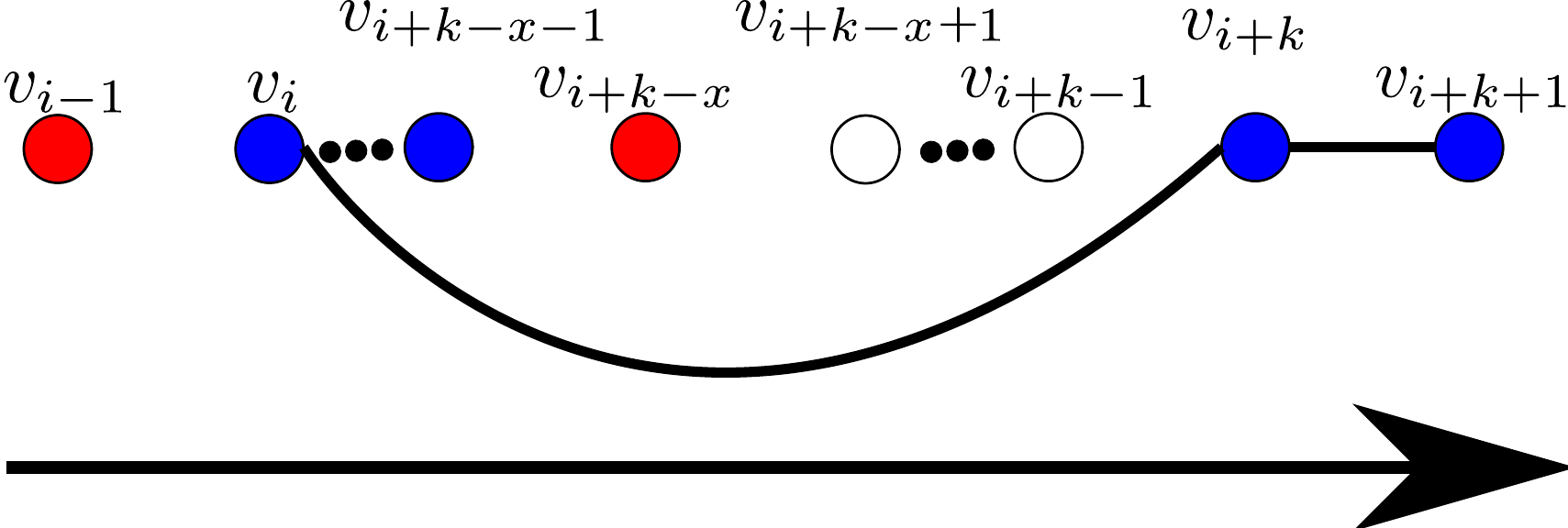}
			\label{fig:catarina2v2}
		}
	\qquad
	\subfloat
		[vertices $v_{i-1}$, $v_{i+k-x}$, and $v_{i+k+1}$ induce a monochromatic
		$P_{3}$ with reach $k+2$]
		{
			\includegraphics[scale=0.3]{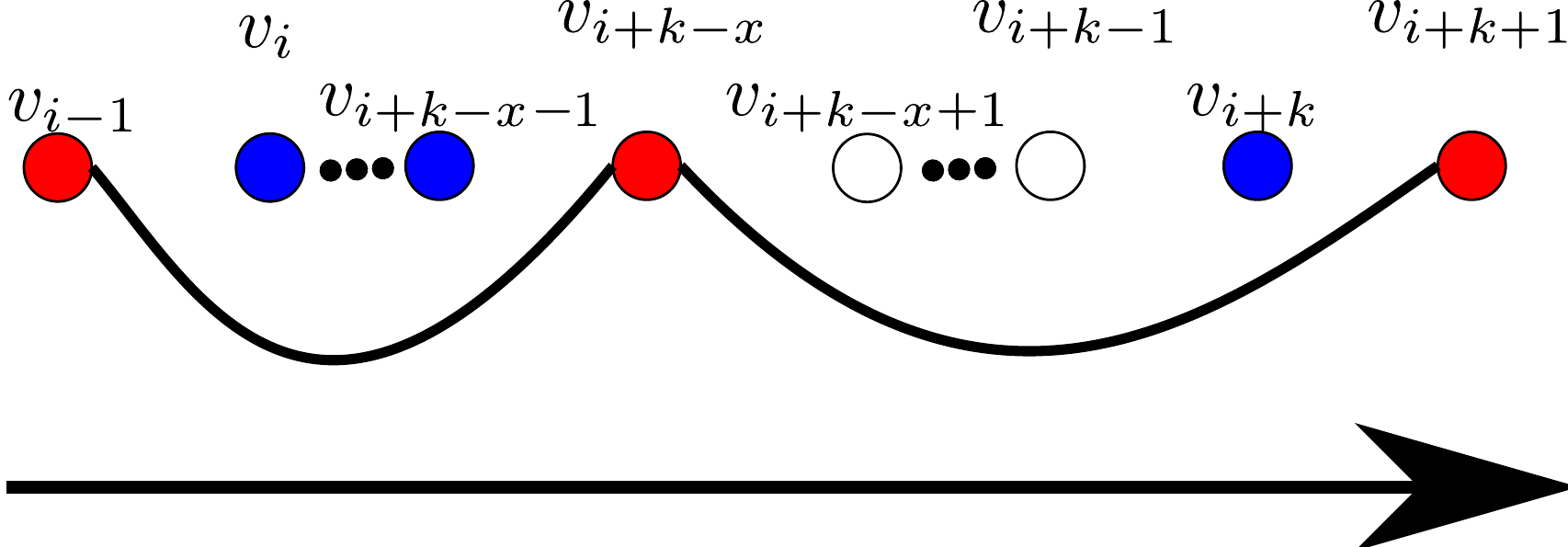}
			\label{fig:catarina3v2}
		}
	\caption{A monochromatic-block of size $x \neq k, k+1$ in a power of a cycle
	$C_{n}^{k}$, with $n \geq 2k + 2$, implies a monochromatic $P_{3}$ with reach
	$k+1$ or $k+2$.}
	\label{fig:catarinav2}
\end{figure}

\begin{lemma}
\label{l:iff}
Let $G$ be a power of a cycle $C_n^k$, where $n \geq 2k + 2$, and consider a
2-colouring of its vertices. If every monochromatic-block has size $k$ or
$k+1$, then $G$ has \textbf{no} monochromatic $P_3$. Otherwise, if
\textbf{not} every monochromatic-block has size $k$ or $k+1$, then $G$ has
a monochromatic $P_3$ with reach $k+1$ or $k+2$; in particular, when
$n = 3k + 2$, $G$ has a monochromatic~$P_3$ with reach $k+1$ or
$G$ has a monochromatic $C_4$.
\end{lemma}

\begin{proof}
Let $G$ be a power of a cycle $C_{n}^{k}$ with $n \geq 2k + 2$.
Consider a 2-colouring $\pi$ of the vertices of $G$ such that every
monochromatic-block has size $k$ or $k+1$. 

Consider any three vertices $v_i$, $v_j$ and $v_\ell$ with the
same colour. Then, either they are in the same monochromatic-block --- and
induce a triangle --- or two of them have indices that differ by at least $k+1$
with respect to the third vertex --- and the three vertices induce a
disconnected graph. In both cases, $v_i$, $v_j$ and $v_\ell$ do not induce a
$P_3$. Hence, no $P_3$ is monochromatic.

Now, consider a 2-colouring $\pi$ of the vertices of $G$ such that there exists
a monochromatic-block of size $x \neq k, k+1$.
Consider a monochromatic-block of size $p \geq k+2$ with vertices
$v_{i}$, $v_{i+1}$, $v_{i+2}$, $\ldots$, $v_{i+k+1}$, $\ldots$, and
$v_{i+p-1}$.
Notice that vertices $v_{i}$, $v_{i+1}$, and $v_{i+k+1}$ induce a~$P_3$. So, we
may assume that there exists a monochromatic-block with vertices
$v_{i}$, $v_{i+1}$, $v_{i+2}$, $\ldots$, $v_{i+k+1}$, $\ldots$, $v_{i+k-x-1}$,
where $x > 0$. By symmetry, consider that $v_i$ has blue colour. Notice that
vertices $v_{i-1}$ and $v_{i+k-x}$ are adjacent and with red colour. Please
refer to Figure~\ref{fig:catarinav2}. Suppose that vertex $v_{i+k}$ has
red colour. Then, vertices $v_{i-1}$, $v_{i+k-x}$, and $v_{i+k}$ induce a monochromatic
$P_3$ with reach $k+1$ (see Figure~\ref{fig:catarina1v2}).
Now, consider vertex $v_{i+k}$ has blue colour. Suppose that vertex $v_{i+k+1}$
has blue colour, then vertices $v_{i+k}$, $v_{i+k+1}$, and $v_{i}$ induce a 
monochromatic $P_3$ with reach $k+1$ (see Figure~\ref{fig:catarina2v2}). Now,
consider vertex $v_{i+k+1}$ has red colour and vertices $v_{i-1}$,
$v_{i+k-x}$, and $v_{i+k+1}$ induce a monochromatic $P_3$ with reach $k+2$ (see
Figure~\ref{fig:catarina3v2}).

Now, consider the case $n = 3k + 2$. We know that $G$ has a monochromatic $P_3$
of reach $k+1$ or $k+2$. In the first case, we are done, so we assume that $G$
has a monochromatic $P_3$ $v_{i-1}$, $v_{i+k-x}$, and $v_{i+k+1}$ of red colour.
Moreover, vertex $v_i$ (resp. vertex $v_{i+k}$) has blue colour, otherwise
vertices $v_i$, $v_{i+k-x}$, and $v_{i + k + 1}$ (resp. vertices $v_{i-1}$,
$v_{i+k-x}$, and $v_{i + k}$) would induce a monochromatic $P_3$ with reach
$k+1$. Vertices $v_{i-1}$, $v_{i+k-x}$, $v_{i+k+1}$, and $v_{i+2k+1}$ induce the
unique $C_4$ that includes vertices $v_{i-1}$, $v_{i+k-x}$, and $v_{i+k+1}$.
Please refer to Figure~\ref{fig:catarinav2.2}.
Suppose vertex $v_{i+2k+1}$ has red colour, then vertices $v_{i-1}$,
$v_{i+k-x}$, $v_{i+k+1}$, and $v_{i+2k+1}$ induce a monochromatic $C_4$ (see
Figure~\ref{fig:catarina4v2}).
Now, consider vertex $v_{i+2k+1}$ has blue colour. Suppose that vertex
$v_{i+2k}$ (resp.$v_{i+2k+2}$) has blue colour, then vertices $v_{i+k}$,
$v_{i+2k}$, and $v_{i+2k+1}$ (resp. $v_{i+2k+1}$, $v_{i+2k+2}$, and
$v_{i+3k+2}$) induce a monochromatic $P_3$ with reach $k+1$ (see
Figure~\ref{fig:catarina5v2}). Now, consider vertices $v_{i+2k}$ and
$v_{i+2k+2}$ have red colour. Vertices $v_{i+k+1}$, $v_{i+2k}$, and
$v_{i+2k+2}$ induce a monochromatic $P_3$ with reach $k+1$ (see
Figure~\ref{fig:catarina6v2}).
\end{proof}

\begin{figure}[t]
\centering
	\subfloat
		[vertices $v_{i-1}$, $v_{i+k-x}$, $v_{i+k+1}$, and $v_{i+2k+1}$ induce a
		monochromatic $C_{4}$] {
			\includegraphics[scale=0.3]{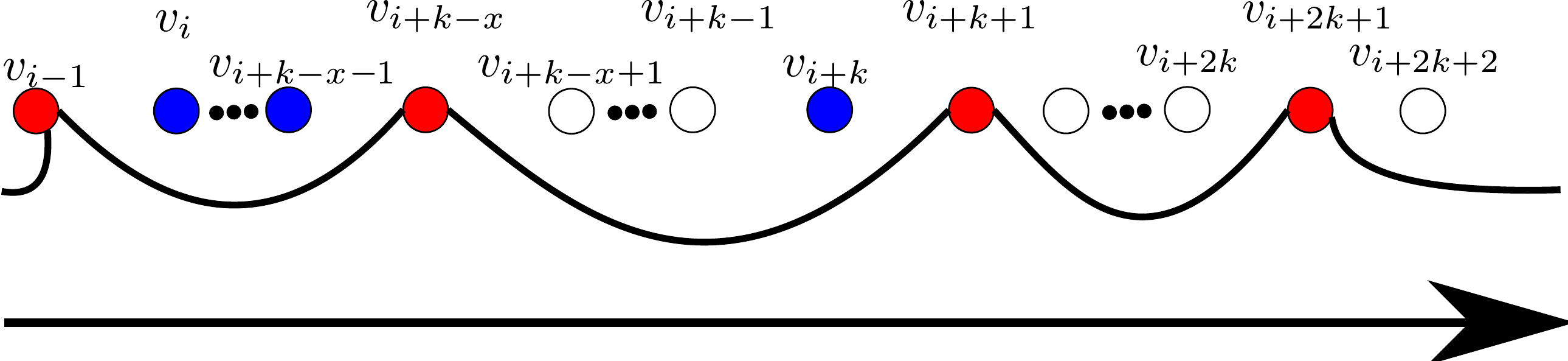}
			\label{fig:catarina4v2}
		}
	\qquad
	\subfloat
		[vertices $v_{i+k}$, $v_{i+2k}$, and $v_{i+2k+1}$ (resp. $v_{i+2k+1}$,
		$v_{i+2k+2}$, and $v_{i}$) induce a monochromatic $P_{3}$ with reach
		$k+1$] {
			\includegraphics[scale=0.3]{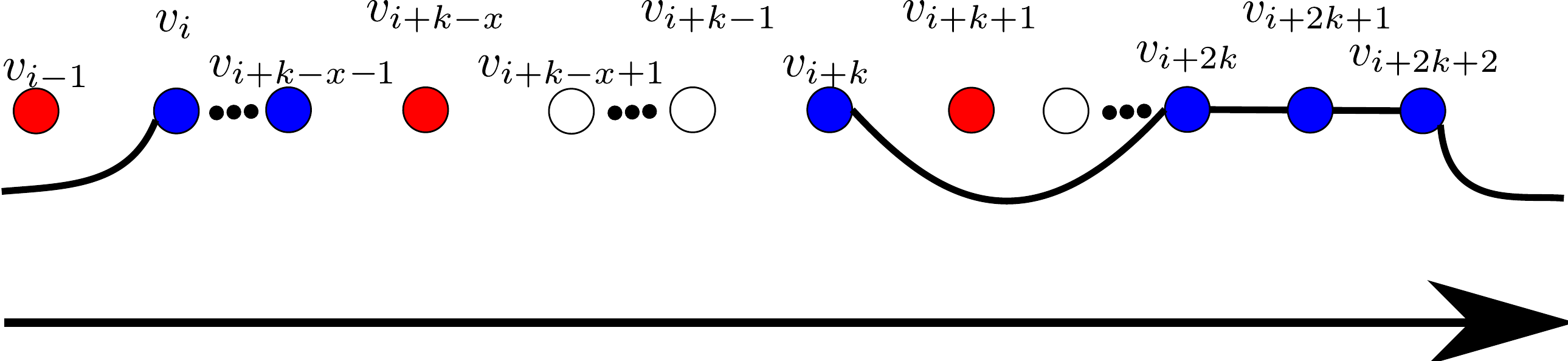}
			\label{fig:catarina5v2}
		}
	\qquad
	\subfloat
		[vertices $v_{i+k+1}$, $v_{i+2k}$, and $v_{i+2k+2}$ induce a monochromatic
		$P_{3}$ with reach $k+1$]
		{
			\includegraphics[scale=0.3]{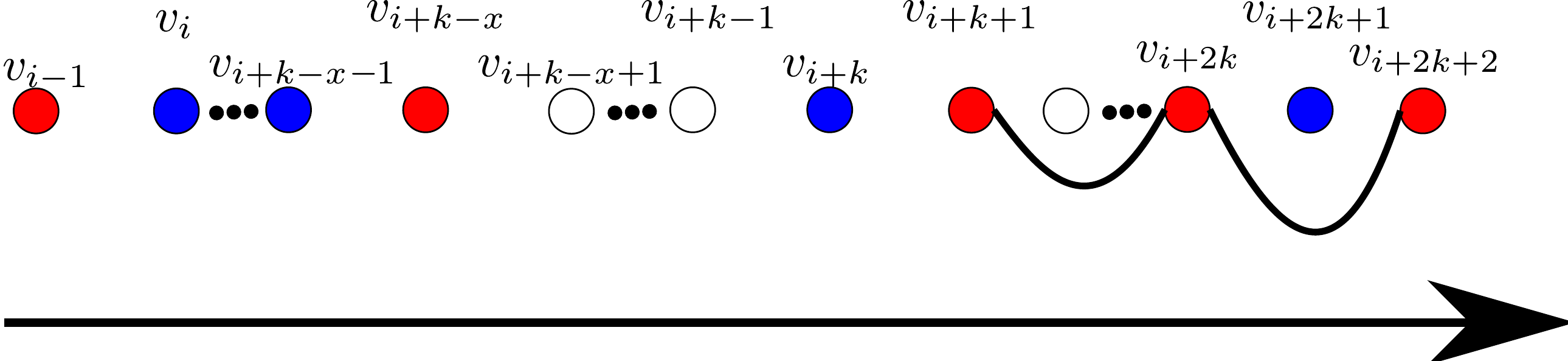}
			\label{fig:catarina6v2}
		}
	\caption{A monochromatic-block of size $x \neq k, k+1$ in a power of a cycle
	$C_{n}^{k}$, with $n = 3k + 2$, implies a monochromatic $P_{3}$ with reach
	$k+1$ or a monochromatic $C_4$.}
	\label{fig:catarinav2.2}
\end{figure}

\begin{theorem}
\label{thm:kappabpowerofcyclethirdinterval}
A power of a cycle $C_n^k$, when $n \geq 3k + 2$, has biclique-chromatic
number~2 if, and only if, there exist natural numbers $a$ and $b$, such that $n =
ak + b(k+1)$ and $a + b \geq 2$ is even.
\end{theorem}

\begin{proof}
Let $G$ be a power of a cycle $C_n^k$ with $n \geq 3k + 2$.

First, consider natural numbers $a$ and $b$, such that $n = ak + b(k+1)$ and $a
+ b \geq 2$ is even. Then, there exists a 2-colouring $\pi$ such that every
monochromatic-block has size $k$ or $k+1$. Lemma~\ref{l:iff} says that $G$ has 
\textbf{no} monochromatic $P_{3}$ and therefore $\pi$ is a 2-biclique-colouring.
We refer to Figure~\ref{fig:semresto} to illustrate such 2-biclique-colouring.

For the converse, suppose that there are no such $a$ and $b$, which implies
that any 2-colouring $\pi^\prime$ of the vertices of $G$ is such that there
exists a monochromatic-block of size $x \neq k, k+1$. Consider $n = 3k +
2$. Lemma~\ref{l:iff} says that such 2-colouring of the vertices of $G$ has
a monochromatic~$P_{3}$ with reach $k+1$ or a monochromatic~$C_{4}$. Every
$P_3$ with reach $k+1$ is a biclique and every $C_4$ is a biclique, which
implies that $\pi^\prime$ is not a 2-biclique-colouring, which is a contradiction.
Now, consider $n > 3k + 2$. Lemma~\ref{l:iff} says that such 2-colouring of
the vertices of $G$ has a monochromatic~$P_{3}$ with reach $k+1$ or $k+2$. Every
$P_3$ with reach $k+1$ or $k+2$ is a $P_3$ biclique, which implies that
$\pi^\prime$ is not a 2-biclique-colouring, which is a contradiction.
\end{proof}

 There exists an efficient algorithm that verifies if the system of equations of
 Theorem~\ref{thm:kappabpowerofcyclethirdinterval} has a solution. If so,
 it also computes values of $a$ and $b$ -- the proof of
 Theorem~\ref{thm:algoritmopowerofcycle} yields
 Algorithm~\ref{alg:numerobicliquecromaticopowerofcycle} to determine if the
 biclique-chromatic number is 2 or 3 and also computes values of $a$ and $b$.
 When the biclique-chromatic number is 2, we define a 2-biclique-colouring
 $\pi:V(G)\rightarrow\{blue,red\}$ as follows. A number $a$ of monochromatic-blocks
 of size $k$ plus a number $b$ of monochromatic-blocks of size $k + 1$ switching
 colours \emph{red} and \emph{blue} alternately. We refer to
 Figure~\ref{fig:semresto} to illustrate the given 2-biclique-colouring.

\begin{theorem}
\label{thm:algoritmopowerofcycle}
 There exists an algorithm that computes the biclique-chromatic
 number of a power of a cycle $C_n^k$, when $n \geq 3k + 2$.
\end{theorem}

\begin{proof}

Theorem~\ref{thm:bicliqueupperboundpowerofcycle} states that the
biclique-chromatic number of a power of a cycle $C_n^k$ is at most~3 and
Theorem~\ref{thm:kappabpowerofcyclethirdinterval} states that a power of a cycle
$C_n^k$ with $n \geq 3k + 2$ has biclique-chromatic number~2 if, and only if,
there exist natural numbers $a$ and $b$, such that $n = ak + b(k+1)$ and $a + b
\geq 2$ is even.

Let $c = a + b$. We show that there exist natural numbers $b$ and $c$, such that
$n = ck + b$, $b \leq c$, and $c$ is even if, and only if, natural numbers $c_0 =
\left\lfloor\frac{n}{k}\right\rfloor$ and $b_0 = n - c_0 k$ have the
following properties: $c_0$ is even and $b_0 \leq c_0$; or 
natural numbers $c_1 = \left\lfloor\frac{n}{k}\right\rfloor - 1$ and $b_1 = n -
c_1 k$ have the following properties: $c_1$ is even and $b_1 \leq c_1$. 

Clearly, $b_0$ and $c_0$ (resp. $b_1$ and $c_1$) are natural numbers such that
$n = c_0 k + b_0$ (resp. $n = c_1 k + b_1$), $b_0 \leq c_0$ (resp. $b_1 \leq
c_1$), $c_0$ (resp. $c_1$) is even, and $c_0 \geq 2$ (resp. $c_1 \geq 2$) since $n
\geq 2k + 2$.

For the converse, suppose that there exist natural numbers $a$ and $b$, such
that $n = ck + b$ and $c$ is even. Let $b^\prime = b$ and $c^\prime = c$. While $b^\prime
\geq 2k$, do $c^\prime := c^\prime + 2$ and $b^\prime := b^\prime - 2k$.
Clearly, in the end of the loop, we have $c^\prime$ even, $b^\prime \geq 0$, and
$c^\prime \geq b^\prime$. Moreover, we consider two cases.

\begin{itemize}
  \item $b^\prime < k$ in the end of the loop. Then, $c^\prime =
  \left\lfloor\frac{n}{k}\right\rfloor$ and $b^\prime = n - c^\prime k$.
  \item $k \leq b^\prime < 2k$ in the end of the loop. Then, $c^\prime =
  \left\lfloor\frac{n}{k}\right\rfloor - 1$ and $b^\prime = n - c^\prime k$.
\end{itemize}
\end{proof}

As a remark, in Theorem~\ref{thm:algoritmopowerofcycle},
we let $c = a + b$ and rewrite the equation $n = ak + b(k+1)$ as $n = ck + b$,
very similar to the Division Algorithm formula. Nevertheless, there is a rather
subtle difference: in the Division Algorithm formula, the choice for the
value of the remainder is bounded by the value of the divisor, while in the
equation $n = ck + b$, the choice for the value of the remainder is bounded by
the choice for the value of the quotient (recall $b \leq c$). This subtle
difference may change drastically the behavior of the equation. More precisely,
given two natural numbers $n$ and $k$, with $n \geq 2k + 2$, it is not
necessarily true that there exist natural numbers $b$ and $c$ such that $n = ck
+ b$, $c \geq 2$ is even, and $b \leq c$. For instance, there do not exist
natural numbers $b$ and $c$ such that $11 = 3c + b$, $c \geq 2$ is even, and $b
\leq c$.

	\begin{algorithm}[t]
	\label{alg:numerobicliquecromaticopowerofcycle}
	\SetKwInOut{Input}{input}\SetKwInOut{Output}{output}
	\Input{$C_n^k$, a power of a cycle with $n \geq 3k + 2$}
	\Output{$\kappa_{B}(C_n^k)$, the biclique-chromatic number of $C_n^k$.}
	\caption{To compute the biclique-chromatic number of a power of a cycle $C_n^k$
	with $n \geq 3k + 2$}
	\BlankLine
	\Begin
	{	
		$c \longleftarrow \left\lfloor \frac{n}{k} \right\rfloor$\;
		$b \longleftarrow n - ck$\;
		\eIf{$c \bmod 2 = 0$ and $c \geq b$}
		{
				\Return{$2$\;}
		}
		{
			$c \longleftarrow \left\lfloor \frac{n}{k} \right\rfloor - 1$\;
			$b \longleftarrow n - ck$\;
			\eIf{$c \bmod 2 = 0$ and $c \geq b$}
			{
					\Return{$2$\;}
			}
			{
				\Return{$3$\;}
			}
		}
	}
	\end{algorithm}

\section{Final considerations}
\label{sec:final}

The reader should notice the structure differences between the two
considered classes of power graphs and observe the similarities on giving lower
and upper bounds on the biclique-chromatic number. For instance,
the lower bound on the biclique-chromatic number in both cases when $n \leq 2k$ 
is a consequence of the existence of a set of $K_2$ bicliques whose union
induces a complete graph --- in the case of powers of cycles, this can happen
only when such union is the whole vertex set, but in the case of powers of
paths such union can be the whole vertex set (when $n \leq k + 1$) or a
vertex subset of size $2k + 2 - n$ (when $k + 2 \leq n \leq 2k$). When $n \geq
2k + 1$, monochromatic-blocks are the key step to construct optimal colourings. 
Nevertheless, in the given colourings, for powers of
paths, vertices $v_0$ and $v_{n-1}$ may have the same colour, which is not the
case for powers of cycles.

Table~\ref{t:tabela} highlights the exact values for the biclique-chromatic
number of the power graphs settled in this work. In
Figures~\ref{fig:kappaboscilapath}~and~\ref{fig:kappaboscilacycle}, we illustrate
the biclique-chromatic number for a fixed value of $k$ and an increasing $n$ of
powers of paths and powers of cycles, respectively.

\begin{figure}[t]
\centering
	\includegraphics[width=7cm]{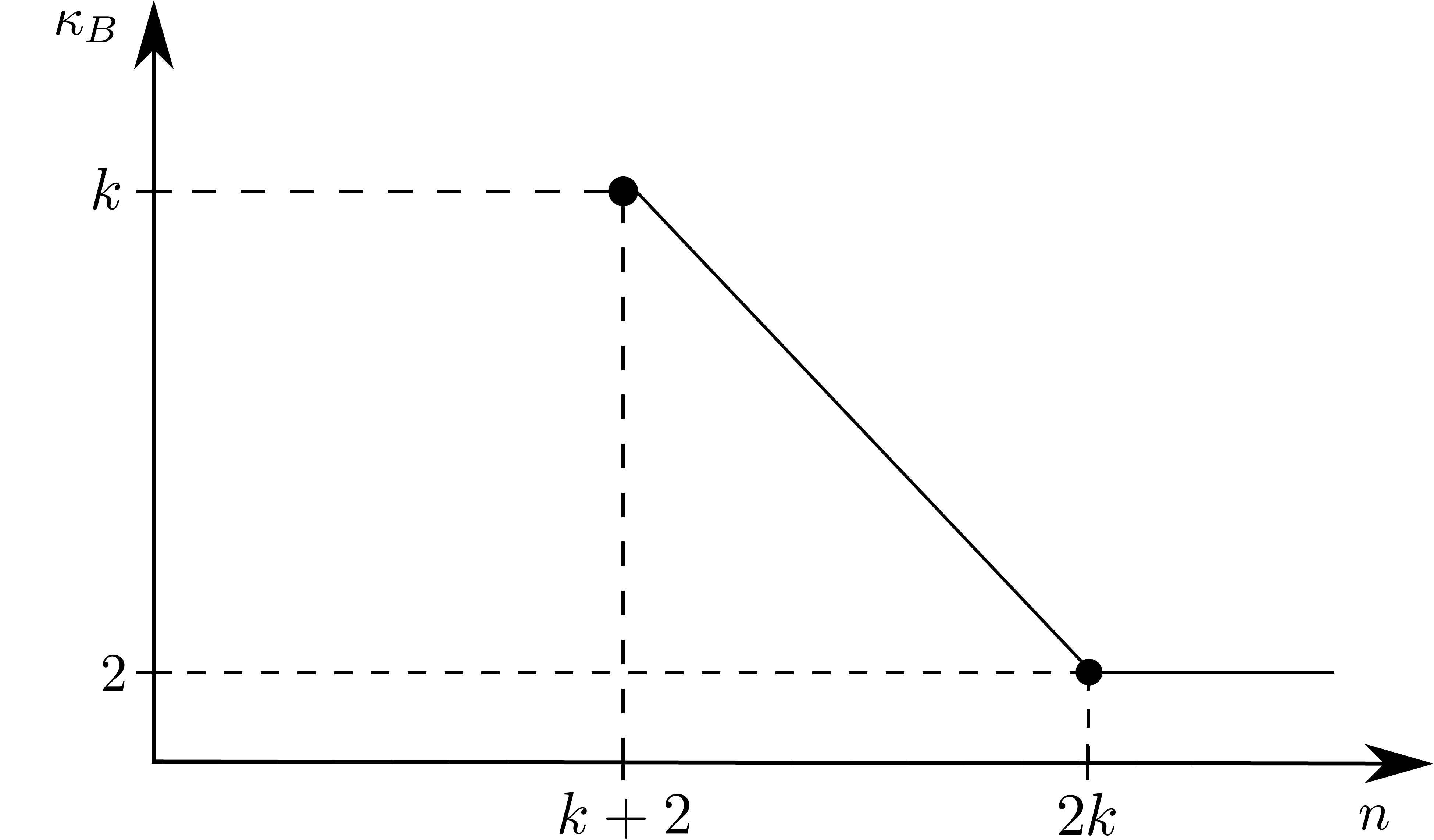}
	\caption{The biclique-chromatic number of a non-complete power of a path for a fixed value of
	$k$ and an increasing $n$}
	\label{fig:kappaboscilapath}
\end{figure}

\begin{figure}[t]
\centering
	\includegraphics[width=10cm]{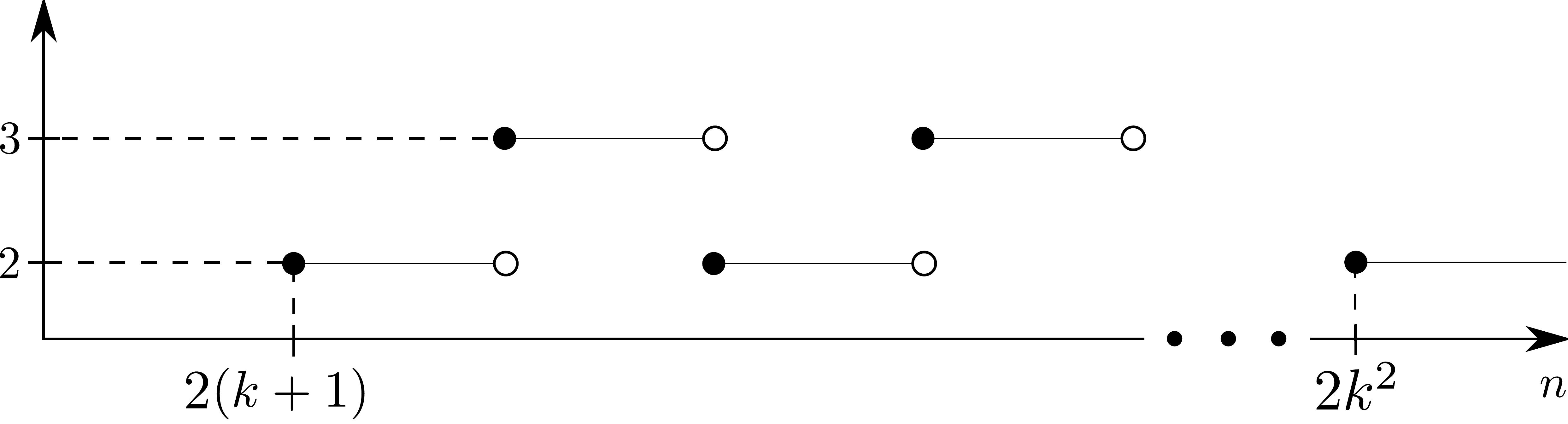}
	\caption{The biclique-chromatic number of a non-complete power of a cycle for a fixed value
	of $k$ and an increasing $n$}
	\label{fig:kappaboscilacycle}
\end{figure}

As a corollary of Theorem~\ref{thm:kappabpowerofcyclethirdinterval}, every non-complete
power of a cycle $C_{n}^{k}$ with $n \geq 2k^2$ has biclique-chromatic number~2.
Thus, the biclique-chromatic number of a power of a cycle $C_{n}^{k}$, for a fixed
value of $k$ and an increasing $n \geq 3k + 2$, does not oscillate forever.

\begin{corollary}
A non-complete power of a cycle $C_{n}^{k}$ with $n \geq 2k^2$ has biclique-chromatic
number~2. 
\end{corollary}
\begin{proof}
Theorem~\ref{thm:division} says that $n = a^\prime k + t$ for natural numbers
$a^\prime $ and $t$, $a^\prime  \geq 2$ is even, and $0 \leq t < 2k$. If we can rewrite $n =
a k + b(k+1)$ with natural numbers $a$ and $b$, such that $a + b \geq 2$ is
even, then Theorem~\ref{thm:kappabpowerofcyclethirdinterval} says that a power of a cycle
$C_n^k$ with $n \geq 2k^2$ has biclique-chromatic number~2. Since $0 \leq t \leq
2k$, $n \geq 2k^2$, and $a^\prime$ is an even natural number, we have 

\begin{eqnarray*}
 n = a^\prime k + t &\geq& 2k^2     \nonumber \\
     a^\prime k &\geq& 2k^2 - 2k + 1     \nonumber \\
     a^\prime &\geq& 2k - 1     \nonumber \\
     a^\prime &\geq& 2k     \nonumber \\
\end{eqnarray*}

Let $a = a^\prime - t$ and $b = t$. Clearly, $a$ and $b$ are natural numbers. 
Moreover, $a + b \geq 2$ is even.
\end{proof}

Groshaus, Soulignac, and Terlisky have recently proposed a related hypergraph
colouring, called \emph{star-colouring}~\cite{1210.7269}, defined as
follows. A \emph{star} is a maximal set of vertices that induces a
complete bipartite graph with a universal vertex and at least one edge. 
The definition of star-colouring follows the same line as clique-colouring and
biclique-colouring: a \emph{star-colouring} of a graph $G$ is a function that
associates a colour to each vertex such that no star is monochromatic. The 
\emph{star-chromatic number} of a graph $G$, denoted by $\kappa_S(G)$, is the 
least number of colours $c$ for which $G$ has a star-colouring with at most~$c$ 
colours. Many of the results of biclique-colouring
achieved in the present work are naturally extended to star-colouring. Since the
constructed graph of Corollary~\ref{cor:checkbicliquecolouring} is $C_4$-free
and the bicliques in a $C_4$-free graph are precisely the stars of the graph,
we can restate Corollary~\ref{cor:checkbicliquecolouring} as follows below.

\begin{corollary}
Let $G$ be a $\{C_4, K_4\}$-free graph. It is co$\mathcal{NP}$-complete to
check if a colouring of the vertices of $G$ is a star-colouring.
\end{corollary}

About star-colouring and the investigated classes of power graphs,
we also have some few remarks. On one hand, the bicliques of a power of a path $P_n^k$ are
the stars of the graph and, consequently, all results obtained for
biclique-colouring powers of paths hold to star-colouring powers of paths.
On the other hand, a power of a cycle $C_n^k$ is not necessarily $C_4$-free, and
there are examples of powers of cycles with $P_3$ stars that are not bicliques
due to the fact that such $P_3$ stars are contained in $C_4$ bicliques of the
graph. This happens for instance in the case $n \in [2k+2, 3k+1]$ and one such
example is graph $C_{11}^4$ exhibited in Figure~\ref{fig:c114}. Notice that the highlighted
vertices form a monochromatic $P_3$ star, so that the colouring is not a
2-star-colouring. The three highlighted vertices together with vertex $u$, on
the other hand, form a polychromatic $C_4$ biclique --- indeed, the
exhibited colouring is a 2-biclique-colouring. 
We summarize the results about star-colouring powers of paths and powers of
cycles in the following theorems and also in Table~\ref{t:tabela}. Please refer
to the line of the table where we consider a power of a cycle with $n \in
[2k+2, 3k+1]$ to check the difference between the biclique-chromatic number
(which is always 2) and the star-chromatic number (which depends on $n$ and
$k$).

\begin{figure}[t]
\centering
	\includegraphics[scale=0.2]{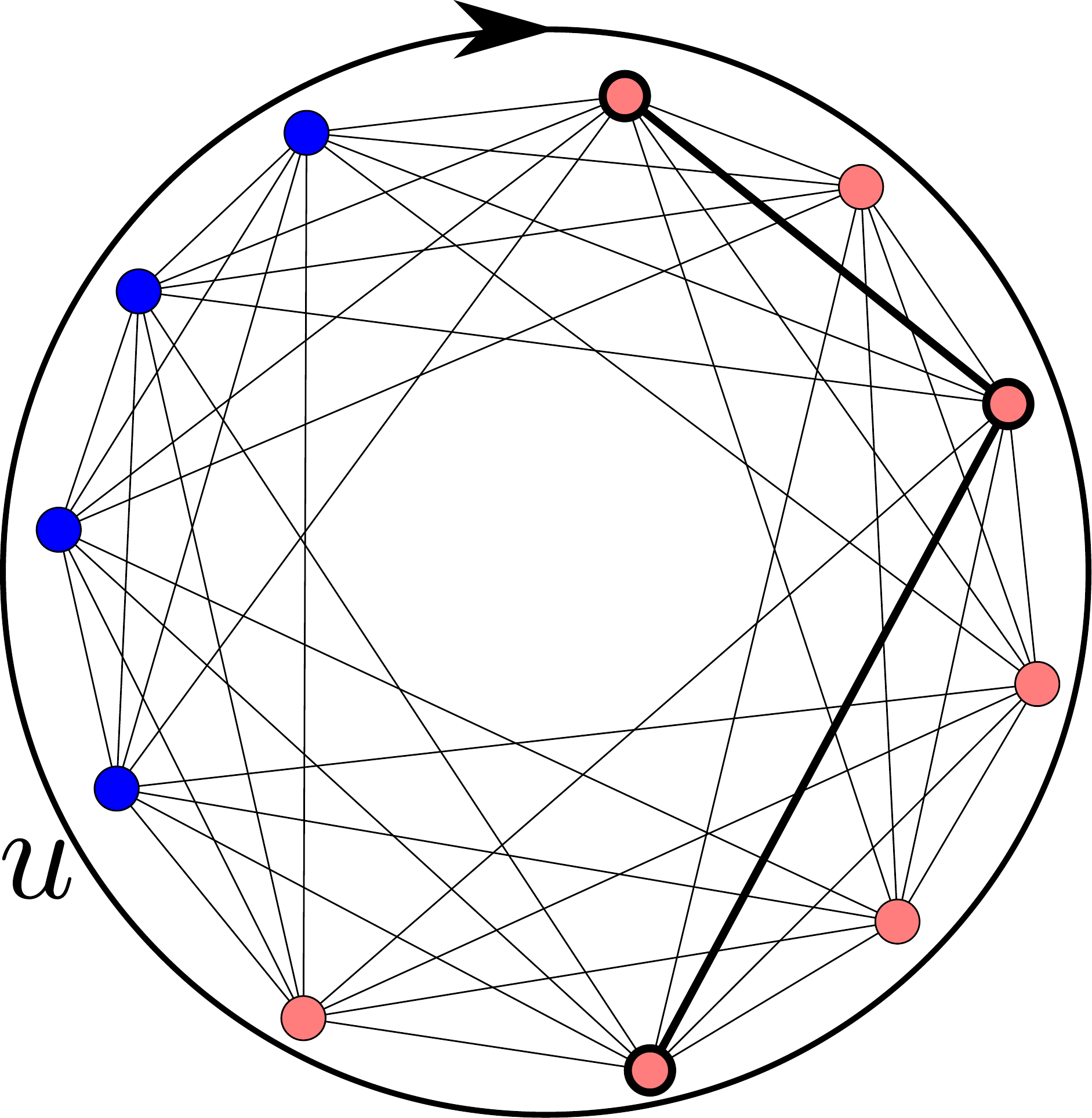}
	\caption{Power of a cycle $C_{11}^{4}$ with a
	2-biclique-colouring which is not a 2-star-colouring. 
	Notice that there exists a monochromatic $P_{3}$ star
	highlighted in bold.}
	\label{fig:c114}
\end{figure}

\begin{theorem}
For any power of a path, the star-chromatic number is equal to the
biclique-chromatic number.
\end{theorem}

\begin{theorem}
A power of a cycle $C_n^k$, when $n \leq 2k + 1$ or $n \geq 3k+2$, has
star-chromatic number equal to the biclique-chromatic number.
If $2k + 2 \leq n \leq 3k + 1$, then $C_n^k$ has star-chromatic number~2 if, and
only if, there exist natural numbers $a$ and $b$, such that $n = ak + b(k+1)$ and
$a + b \geq 2$ is even. If there does not exist such natural numbers, it has
star-chromatic number~3.
\end{theorem}

\begin{table}[h]
\begin{center}
\begin{tabular}{|c||l|p{3cm}|p{3cm}|}
\hline
  Graph $G$ & Range of $n$ & $\kappa_B(G)$ & $\kappa_S(G)$ \\ \hline\hline
  \multirow{3}{*}{$P_{n}^{k}$} & $[1, k + 1]$ & $n$ & $n$ \\
  \cline{2-4}
   & $[k + 2, 2k]$ & $2k + 2 - n$ & $2k + 2 - n$\\
  \cline{2-4}
   & $[2k + 1, \infty[$ & $2$ & $2$ \\
  \hline\hline
  \multirow{5}{*}{$C_{n}^{k}$} & $[1, 2k + 1]$ & $n$ & $n$ \\
  \cline{2-4}
   & $[2k+2, 3k+1]$ & $2$ &\\
  \cline{2-3}
   & \multirow{4}{*}{$[3k + 2, 2k^2[$} & \multicolumn{2}{l|}{$2$, if there
   exist natural numbers $a$ and $b$,} \\
   & & \multicolumn{2}{l|}{such that $n~=~ak~+~b(k+1)$} \\
   & & \multicolumn{2}{l|}{and $a + b \geq 2$ is
   even;} \\
  & & \multicolumn{2}{l|}{$3$, otherwise.} \\
  \cline{2-4}
   & $[2k^2, \infty[$ & $2$ & $2$\\
  \hline
\end{tabular}
\caption{Biclique- and star-chromatic numbers of powers of paths and powers of cycles}
\label{t:tabela}
\end{center}
\end{table}

A \emph{distance graph} $P_n(d_1,\dots, d_k)$ is a simple graph with
$V(G)= \{v_0,\dots, v_{n-1}\}$ and $E(G)=E^{d_1}\cup\dots\cup E^{d_k}$, such
that $\{v_i,v_j\}\in E^{d_\ell}$ if, and only if, it has reach -- in the context
of a power of a path -- $d_\ell$. Notice that a distance graph
$P_n(d_1,\dots,d_k)$ is a power of a path if $d_1=1$, $d_i=d_{i-1}+1$, and
$d_k< n - 1$. A \emph{circulant graph} $C_n(d_1,\dots, d_k)$ has the same
definition as the distance graph, except by the reach, which, in turn, is in
the context of a power of a cycle. Notice that a circulant
graph~$C_n(d_1,\dots,d_k)$ is a power of a cycle if $d_1=1$, $d_i=d_{i-1}+1$,
and $d_k< \lfloor \frac{n}{2} \rfloor$.
Circulant graphs have been proposed for various practical
applications~\cite{circulantgraphapplication}.
We suggest, as a future work, to biclique colour the classes of distance
graphs and circulant graphs, since colouring problems for distance graphs and
for circulant graphs have been extensively
investigated~\cite{MR2567972,MR1900685,MR1632015}.
Moreover, some results of intractability have been obtained, e.g. determining
the chromatic number of circulant graphs in general is an $\mathcal{NP}$-hard
problem~\cite{MR1653503}.

\section*{Acknowledgments}
The authors would like to thank Renan Henrique Finder for the discussions on
the algorithm to compute the biclique-chromatic number of a power of a cycle
$C_n^k$, when $n \geq 3k + 2$; and to thank Vin{\'i}cius Gusm{\~a}o Pereira de
S{\'a} and Guilherme Dias da Fonseca for discussions on the complexity of
numerical problems. At last, but not least, we thank Vanessa Cavalcante for
the careful proofreading of this paper.

\bibliographystyle{plainnat}

\bibliography{ctw-2012-full}

\end{document}